\documentclass[reqno]{amsart}

\usepackage{hyperref}
\usepackage{bbm}
\usepackage{amsfonts,amsmath,amsxtra,amsthm,amssymb,latexsym}
\usepackage{ifpdf}
\usepackage{tikz}
\usepackage{color}

\newcommand{\ack}{\section*{Acknowledgments}}

\newtheorem{theorem}{Theorem}[section]
\newtheorem{corollary}[theorem]{Corollary}
\newtheorem{lemma}[theorem]{Lemma}

\newtheorem{remark}[theorem]{Remark}
\newtheorem{proposition}[theorem]{Proposition}
\newtheorem{hypothesis}{Hypothesis}[section]
\theoremstyle{definition}
\newtheorem{definition}[theorem]{Definition}
\newtheorem{example}[theorem]{Example}

\newcommand{\nn}{\nonumber}
\newcommand{\be}{\begin{equation}}
\newcommand{\ee}{\end{equation}}
\newcommand{\ba}{\begin{array}}
\newcommand{\ea}{\end{array}}

\newcommand{\ti}{\tilde}

\numberwithin{equation}{section}

 \DeclareMathOperator{\im}{Im}
\DeclareMathOperator{\re}{Re} 
 \DeclareMathOperator{\dom}{dom}
\DeclareMathOperator{\ran}{ran} \DeclareMathOperator{\Ext}{Ext}
\DeclareMathOperator{\Span}{span}\DeclareMathOperator{\ess}{ess}

\DeclareMathOperator{\ac}{ac}
\DeclareMathOperator{\diag}{diag}

\newcommand\R{{\mathbb{R}}}
\newcommand\C{{\mathbb{C}}}

\newcommand\N{{\mathbb{N}}}
\newcommand\Z{{\mathbb{Z}}}

\newcommand\bO{{\mathbb{O}}}

\newcommand\gH{{\mathfrak{H}}}
\newcommand\gS{{\mathfrak{S}}}

\newcommand{\gG}{{\Gamma}}

\newcommand{\gt}{\mathfrak{t}}

\newcommand{\gA}{{\alpha}}

\newcommand\cH{{\mathcal{H}}}

\newcommand\cT{{\mathcal{T}}}
\newcommand\cL{{\mathcal{L}}}
\newcommand\cG{{\mathcal{G}}}
\newcommand\cE{{\mathcal{E}}}
\newcommand\cP{{\mathcal{P}}}
\newcommand\cV{{\mathcal{V}}}
\newcommand\cC{{\mathcal{C}}}
\newcommand\cB{{\mathcal{B}}}
\newcommand\cN{{\mathcal{N}}}
\newcommand\OO{{\mathcal{O}}}

\newcommand\bH{{\mathbf{H}}}
\newcommand\rh{{\mathbf{h}}}

\newcommand\rH{{\rm{H}}}
\newcommand\E{{\rm{e}}}

\newcommand\rR{{\rm{R}}}
\newcommand\rQ{{\rm{Q}}}
\newcommand\rD{{\rm{d}}}
\newcommand\I{{\rm{i}}}
\newcommand\rI{{\rm{I}}}
\newcommand\op{{\rm{op}}}

\newcommand\tr{{\rm{tr}}}
\newcommand\Deg{{\rm{Deg}}}

\newcommand\g{{\bf{g}}}
\newcommand\f{{\bf{f}}}

\def\mul{{\rm mul\,}}
\def\wt#1{{{\widetilde #1} }}

\newcommand*{\mailto}[1]{\href{mailto:#1}{\nolinkurl{#1}}}
\newcommand{\arxiv}[1]{\href{http://arxiv.org/abs/#1}{arXiv:#1}}
\newcommand{\doi}[1]{\href{http://dx.doi.org/#1}{doi: #1}}
\newcommand{\msc}[1]{\href{http://www.ams.org/msc/msc2010.html?t=&s=#1}{#1}}

\begin{document}

\title[Infinite Quantum Graphs]{Spectral Theory of Infinite Quantum Graphs}

\dedicatory{Dedicated to the memory of M.\ Z.\ Solomyak (16.05.1931 -- 31.07.2016)}

\author[P.\ Exner]{Pavel Exner}
\address{Doppler Institute for Mathematical Physics and Applied Mathematics\\ Czech Technical University\\
B\v rehov\'a 7\\ 11519 Prague\\ Czechia\\ and
Department of Theoretical Physics\\ Nuclear Physics Institute\\ Czech Academy of Sciences\\ 25068 \v{R}e\v{z} near Prague, Czechia}
\email{\mailto{exner@ujf.cas.cz}}
\urladdr{\url{http://gemma.ujf.cas.cz/~exner/}}

\author[A.\ Kostenko]{Aleksey Kostenko}
\address{Faculty of Mathematics and Physics\\ University of Ljubljana\\ Jadranska 21\\ 1000 Ljubljana\\ Slovenia\\ and Faculty of Mathematics\\ University of Vienna\\
Oskar--Morgenstern--Platz 1\\ 1090 Wien\\ Austria}
\email{\mailto{Aleksey.Kostenko@fmf.uni-lj.si};\ \mailto{Oleksiy.Kostenko@univie.ac.at}}
\urladdr{\url{http://www.mat.univie.ac.at/~kostenko/}}

\author[M.\ Malamud]{Mark Malamud}
\address{Peoples Friendship University of Russia (RUDN University)\\ Miklukho-Maklaya Str. 6\\ 117198 Moscow\\Russia}
\email{\mailto{malamud3m@gmail.com}}

\author[H.\ Neidhardt]{Hagen Neidhardt}
\address{Weierstrass Institute for Applied Analysis and Stochastics\\ Mohrenstr. 39\\ 10117
Berlin\\ Germany}
\email{\mailto{neidhard@wias-berlin.de}}
\urladdr{\url{http://www.wias-berlin.de/~neidhard/}}

\thanks{{\it Research supported by the Czech Science Foundation (GA\v{C}R)
under grant No.~17-01706S and the European Union within the project CZ.02.1.01/0.0/0.0/16\textunderscore 019/0000778 (P.E.), by the Austrian Science Fund (FWF) under grant No.~P28807 (A.K.), by the ``RUDN University Program 5-100"  (M.M.), and by the European Research Council (ERC) under grant No.~AdG 267802 ``AnaMultiScale"\ (H.N.)}}

\thanks{Ann.\ Henri\ Poincar\'e (to appear); \doi{10.1007/s00023-018-0728-9}} 
 
\thanks{\arxiv{1705.01831}} 



\begin{abstract}
We investigate quantum graphs with infinitely many vertices and edges without the common restriction on the geometry of the underlying metric graph that there is a positive lower bound on the lengths of its edges. Our central result is a close connection between spectral properties of a quantum graph and the corresponding properties of a certain weighted discrete Laplacian on the underlying discrete graph. Using this connection together with spectral theory of (unbounded) discrete Laplacians on infinite graphs, we prove a number of new results on spectral properties of quantum graphs. Namely, we prove several self-adjointness results including a Gaffney type theorem. We investigate the problem of lower semiboundedness, prove several spectral estimates (bounds for the bottom of spectra and essential spectra of quantum graphs, CLR-type estimates) and study spectral types.   
\end{abstract} 

\maketitle

{\scriptsize{\tableofcontents}}

\section{Introduction}\label{intro}

During the last two decades, {\em quantum graphs} became an extremely popular subject because of numerous applications in mathematical physics, chemistry and engineering. Indeed, the literature on quantum graphs is vast and extensive and there is no chance to give even a brief overview of the subject here. We only mention a few recent monographs and collected works with a comprehensive bibliography:  \cite{bcfk06},  \cite{bk13}, \cite{ekkst08}, \cite{gs08} and \cite{post}. The notion of quantum graph refers to a graph $\cG$ considered as a one-dimensional simplicial complex and equipped with a differential operator (``Hamiltonian''). The idea has it roots in the 1930s when it was proposed to model free electrons in organic molecules  \cite{paul36,rusch53}. It was rediscovered in the late 1980s and since that time it found numerous applications. Let us briefly mention some of them: superconductivity theory in granular and artificial materials \cite{al83, rusc98}, microelectronics and waveguide theory \cite{exs89, ml71,
mv05}, Anderson localization in disordered wires \cite{asw05, asw06, ehs07}, chemistry (including studying carbon nanostructures)
\cite{alm04, ex01, korl07, kp07, rb72}, photonic crystal theory \cite{akk99, fk96, ku01}, quantum chaotic systems \cite{gs08, ks97}, and others. These applications of quantum graphs usually involve modelling of waves of various nature propagating in thin branching media which looks like a thin neighbourhood $\Omega$ of a graph $\cG$. A rigorous justification of such a graph approximation is a nontrivial problem. It was first addressed in the situation where the boundary of the ``fat graph'' is Neumann (see, e.g., \cite{kuze01, rusc01}), a full solution was obtained only recently \cite{cet10, expo13}. The Dirichlet case is more difficult and a  work remains to be done (see, e.g., \cite{gri08, post}).

From the mathematical point of view, quantum graphs are interesting because they are a good model to study properties of quantum systems depending on geometry and topology of the configuration space. They exhibit a mixed dimensionality being locally one-dimensional but globally multi-dimensional of many different types. To the best of our knowledge, however, their analysis usually  includes the assumption that there is a positive lower bound on the lengths of the graph edges (we are aware only of a few works dealing with metric graphs having edges of arbitrarily small length, however, with some other additional rather restrictive assumptions, e.g., radially symmetric trees \cite{car00, sol02}, some classes of fractals \cite{ar18,arfk18,arkt16}, graphs having finite total length or diameter \cite{car11}). Our main aim is to investigate spectral properties of quantum graphs avoiding this rather restrictive hypothesis on the geometry of the underlying metric graph $\cG$.

To proceed further we need to introduce briefly some notions and structures (a detailed description is given in Section \ref{sec:BT_QG}).
Let $\cG_d = (\cV,\cE)$ be a (combinatorial) graph with finite or countably infinite sets of vertices $\mathcal{V}$ and edges $\mathcal{E}$.  For two different vertices $u$, $v\in \mathcal{V}$ we shall write $u\sim v$ if there is an edge $e\in \mathcal{E}$ connecting $u$ with $v$. For every $v\in \mathcal{V}$, $\cE_v$ denotes the set of edges incident to the vertex $v$.
 To simplify our considerations, we assume that the graph $\cG_d$ is connected and there are no loops or multiple edges (let us mention that these assumptions can be made without loss of generality, see Remark \ref{rem:assumptions} below). In what follows we shall also assume that $\cG_d$ is equipped with a metric, that is, each edge $e\in\mathcal{E}$ is assigned with the length $|e|=l_e\in (0,\infty)$ in a suitable way. A graph $\cG_d$ equipped with a metric $|\cdot|$ is called a {\em metric graph} and is denoted by $\cG=(\cG_d,|\cdot|)$.  Identifying every edge $e$ with the interval $(0,|e|)$ one can introduce the Hilbert space $L^2(\cG) = \bigoplus_{e\in\cE} L^2(e)$ and then the Hamiltonian $\bH$ which acts in this space as the (negative) second derivative $-\frac{d^2}{dx_e^2}$ on every edge $e\in\cE$. To give $\bH$ the meaning of a quantum mechanical energy operator, it must be self-adjoint. To make it symmetric, one needs to impose appropriate boundary conditions at the vertices. Kirchhoff conditions \eqref{eq:kirchhoff} or, more generally, $\delta$-type conditions with interactions strength $\alpha\colon\cV\to \R$
\[
\begin{cases} f\ \text{is continuous at}\ v,\\[1mm] \sum_{e\in \cE_v}f_e'(v) = \alpha(v)f(v), \end{cases}\qquad v\in\cV,
\]
are the most standard ones (cf. \cite{bk13}). Restricting further to functions vanishing everywhere except finitely many edges, we end up with the pre-minimal symmetric operator $\bH_\alpha^0$ (see Section \ref{sec:parameterization} for a precise definition). The first question which naturally appears in this context is, of course, whether this operator is essentially self-adjoint in $L^2(\cG)$ (which is the same that its closure $\bH_{\alpha} = \overline{\bH_\alpha^0}$ is self-adjoint).  
 To the best of our knowledge, in the case when both sets $\cV$ and $\cE$ are countably infinite, the self-adjointness of $\bH_{\alpha}$ was established when $\inf_{e\in\mathcal{E}}|e|>0$ and the interactions strength $\alpha\colon \cV\to\R$ is bounded from below in a suitable sense (see, e.g., \cite[Chapter I]{bk13} and \cite{lsv14}). The subsequent analysis of $\bH_{\gA}$ was then naturally performed under these rather restrictive assumptions on $\cG$ and $\alpha$.

We propose a new approach to investigate spectral properties of infinite quantum graphs. To this goal, we exploit the boundary triplets machinery \cite{DM91, Gor84, schm}, a powerful approach to extension theory of symmetric operators (see Appendix \ref{app:triplets} for further details and references). Consider in $L^2(\cG)$ the following operator
\begin{align}\label{eq:Hmin}
\bH_{\min} & = \bigoplus_{e\in \cE} \rH_{e,\min}, & \rH_{e,\min} & = -\frac{\rD^2}{\rD x_e^2},\quad \dom(\rH_{e,\min}) = W^{2,2}_0(e),
\end{align}
where $W^{2,2}_0(e)$ denotes the standard Sobolev space on the edge $e\in\cE$. Clearly, $\bH_{\min}$ is a closed symmetric operator in $L^2(\cG)$ with deficiency indices ${\rm n}_\pm(\bH_{\min})=2\#(\cE)$. In particular, the deficiency indices are infinite when $\cG$ contains infinitely many edges and hence in this case the description of self-adjoint extensions and the study of their spectral properties is a very nontrivial problem. Despite some skepticism (see, e.g., \cite[p.483]{ekkst08}), we are indeed able to construct a suitable boundary triplet for the maximal operator $\bH_{\max}:=\bH_{\min}^\ast$ in the case when $\inf_{e\in\mathcal{E}}|e|=0$. As an immediate outcome, the boundary triplets approach enables us to parameterize the set of all self-adjoint (respectively, symmetric) extensions of $\bH_{\min}$ in terms of self-adjoint (respectively, symmetric) ``boundary linear relations''. 
Furthermore, it turns out (see Proposition \ref{prop:bo_alpha}) that the boundary relation (to be more precise, its operator part) parameterizing the quantum graph operator $\bH_{\gA}$ is unitarily equivalent to the weighted discrete Laplacian $\rh_\gA$ defined in $\ell^2(\cV;m)$ by the following expression
\be\label{eq:1.03}
(\tau_{\cG,\alpha} f)(v) := \frac{1}{m(v)}\left( \sum_{u\in \cV} b(u,v)(f(v) - f(u))  + \alpha(v)f(v)\right),\quad v\in\cV,
\ee
where the weight functions $m\colon \cV\to (0,\infty)$ and $b\colon \cV\times \cV \to [0,\infty)$ are given by
\begin{align}\label{eq:1.04}
m\colon & v\mapsto \sum_{e\in\cE_v}|e|, & b\colon (u,v) & \mapsto \begin{cases} |e_{u,v}|^{-1}, & u\sim v,\\ 0, & u\not\sim v.\end{cases}
\end{align}
Therefore, spectral properties of the quantum graph Hamiltonian $\bH_\gA$ and the discrete Laplacian $\rh_\gA$ are closely connected. For example, we show that (see Theorem \ref{th:main}): 
\begin{itemize}
\item[(i)] {\em The deficiency indices of $\bH_\alpha$ and $\rh_\alpha$ are equal. 
In particular, $\bH_\alpha$ is self-adjoint if and only if $\rh_\alpha$ is self-adjoint.}
\end{itemize}
Assume additionally that the operator $\bH_\alpha$ (and hence also the operator $\rh_\alpha$) is self-adjoint. Then:
\begin{itemize}
\item[(ii)] {\em $\bH_\alpha$ is lower semibounded if and only if $\rh_\alpha$ is lower semibounded.} 
\item[(iii)] {\em The total multiplicities of negative spectra of $\bH_\alpha$ and $\rh_\alpha$ coincide. In particular, $\bH_\alpha$ is nonnegative if and only if the operator $\rh_\alpha$ is nonnegative. Moreover, negative spectra of $\bH_\alpha$ and $\rh_\alpha$ are discrete simultaneously.}
\item[(iv)] {\em $\bH_\alpha$ is positive definite if and only if  $\rh_\alpha$ is positive definite. }
\item[(v)] {\em If in addition $\rh_\alpha$ is lower semibounded, then $\inf\sigma_{\ess}(\bH_\alpha)>0$ $(\inf\sigma_{\ess}(\bH_\alpha)=0)$ exactly when $\inf\sigma_{\ess}(\rh_\alpha)>0$ $($respectively, $\inf\sigma_{\ess}(\rh_\alpha)=0)$.}
\item[(vi)] {\em The spectrum of $\bH_\alpha$ is purely discrete if and only if the number $\#\{e\in\cE\colon |e|>\varepsilon\}$ is finite for every $\varepsilon>0$ and the spectrum of $\rh_\alpha$ is purely discrete.}
\end{itemize}

Spectral theory of discrete Laplacians on graphs has a long and venerable history due to its numerous applications in probability (e.g., random walks on graphs) and physics (see the monographs \cite{cdv}, \cite{dsv}, \cite{ds}, \cite{lub}, \cite{lp}, \cite{var}, \cite{woe} and references therein). If $\inf_{e\in\cE} |e|=0$, then the corresponding discrete Laplacian $\rh_\alpha$ might be unbounded even if $\alpha\equiv 0$. A significant progress in the study of unbounded discrete Laplacians has been achieved during the last decade (see the surveys \cite{kel15}, \cite{kl10}) and we apply these recent results to investigate spectral properties of quantum graphs in the case when $\inf_{e\in\cE} |e|=0$. For example, using (i), we establish a Gaffney type theorem (see Corollary \ref{cor:gaffney} and Remark \ref{rem:4.10}) 
 by simply applying the corresponding result for discrete operators (see \cite[Theorem 2]{hkmw13}): {\em if $\cG$ equipped with a natural path metric is complete as a metric space, then $\bH_0$ is self-adjoint}. Notice that by a Hopf--Rinow type theorem from \cite{hkmw13}, $(\cV,\varrho_0)$ is complete as a metric space if and only if $\cG$ satisfies the so-called {\em finite ball condition} (see, e.g., \cite[Assumption 1.3.5]{bk13}). 
 Combining (iv) and (v) with the Cheeger type and the volume growth estimates for discrete Laplacians (see \cite{bkw15}, \cite{fo}, \cite{kel15}, \cite{kl12}), we prove several spectral estimates for $\bH_0$. In particular, we prove necessary (Theorem \ref{th:H0positive2}(iii)) and sufficient (Theorem \ref{th:H0positive}(iii)) discreteness conditions for $\bH_0$. In the case $\#\cE=\infty$, it follows from (vi) that the condition $\inf_{e\in\cE}|e|=0$ is necessary for the spectrum of $\bH_0$ to be discrete and this is the very reason why the discreteness problem has not been addressed previously for quantum graphs (perhaps, the only exception is the case of radially symmetric trees since for this class of quantum graphs it is possible to reduce the spectral analysis to the analysis of Sturm--Liouville operators, see \cite[\S 5.3]{sol02}).  

Let us also stress that some of our results are new even if $\inf_{e\in\cE}|e|>0$. In this case the discrete Laplacian $\rh_0$ is bounded and hence we immediately conclude by applying (i) that $\bH_\gA$ is self-adjoint for any $\gA\colon\cV\to \R$ (Corollary \ref{cor:sa02}). On the other hand, $\rh_0$ is bounded if and only if {\em the weighted degree} function $\Deg\colon\cV\to \R$ defined by
\[
\Deg\colon v\mapsto \frac{1}{m(v)}\sum_{u\in\cE_v} b(u,v) = \frac{\sum_{e\in\cE_v} |e|^{-1}}{\sum_{e\in\cE_v} |e|}
\]
is bounded on $\cV$ (see \cite{dav}). Therefore, $\bH_\gA$ is self-adjoint for any $\gA\colon\cV\to \R$ in this case too (Lemma \ref{lem:sa01}). Let us stress that the condition $\inf_{e\in\cE}|e|>0$ is sufficient for $\Deg$ to be bounded on $\cV$, however, it is not necessary (see Example \ref{ex:tree01}). 

The duality between spectral properties of continuous and discrete operators on finite graphs and networks was observed by physicists in the 1960s and then by mathematicians in the 1980s \cite{vB,cat,ex97,ni85,roth}.  For a particular class, the so-called equilateral graphs, it is even possible to prove a sort of unitary equivalence between continuous and discrete operators \cite{bgp07,pan12,pan13} (actually, this can also be viewed as the analog of Floquet theory for periodic Sturm--Liuoville operators, cf. \cite{abmn05}). However, in all those cases $\inf_{e\in\cE}|e|>0$ is satisfied and the corresponding difference Laplacian in contrast to \eqref{eq:1.03} is given by
\be\label{eq:1.03B}
(\tau f)(v) := \frac{1}{\deg(v)}\left(\sum_{u\sim v} \frac{f(v) - f(u)}{|e_{u,v}|} + \alpha(v)f(v)\right),\qquad v\in\cV,
\ee
that is, the weight function $m$ is replaced by the combinatorial degree function (see, e.g., \cite[Chapter II]{post}, \cite{roso10}). These functions coincide only if the graph is equilateral and then both \eqref{eq:1.03} and \eqref{eq:1.03B} (with $\alpha\equiv 0$) reduce to the combinatorial Laplacian on $\cG_d$. Moreover, in the case $\inf_{e\in\E}|e|=0$, spectral properties of operators defined by \eqref{eq:1.03} and \eqref{eq:1.03B} can completely be different and spectral properties of \eqref{eq:1.03B} do not correlate with those of the quantum graph operator $\bH_\gA$ (this can be seen by  simple examples of Jacobi matrices, see Remark \ref{rem:difJacobi}).

In fact, it is not difficult to discover certain connections just by considering the corresponding quadratic forms. Namely, let $f$ be a continuous compactly supported function on the metric graph $\cG$, which is linear on every edge. Setting $f_{\cV}:=f\upharpoonright_\cV$, we then get (see Remark \ref{rem:3.8} for more details)
\be\label{eq:1.05}
\begin{split}
\gt_{\bH_\gA}[f]:=\big(\bH_{\gA}f,f\big)_{L^2(\cG)} &=\frac{1}{2}\sum_{u,v\in \cV} b(v,u)|f(v) - f(u)|^2 + \sum_{v\in\cV}\alpha(v)|f(v)|^2\\
&= \big(\rh_\gA f_{\cV},f_{\cV}\big)_{\ell^2(\cV;m)}=:\gt_{\rh_\gA}[f_\cV].
\end{split}
\ee
If $\gA\colon \cV\to [0,\infty)$, then the closures of both forms $\gt_{\bH_\gA}$ and $\gt_{\rh_\gA}$ are regular Dirichlet forms whenever the corresponding graph $\cG$ is locally finite (cf. \cite{fuk10}). 
 Clearly, \eqref{eq:1.05} establishes a close connection between the corresponding Markovian semigroups as well as between Markov processes on the corresponding graphs. However, let us stress that it was exactly the above statement (iii) which helped us to improve and complete one result of G.\ Rozenblum and M.\ Solomyak \cite{roso10} on the behavior of the heat semigroups generated by $\bH_0$ and $\rh_0$ (see Theorem \ref{lem:rozsol} and Remark \ref{rem:rozsol}): {\em for $D>2$ the following equivalence holds}
\[
\|\E^{-t\bH_0}\|_{L^1\to L^\infty} \le  C_1t^{-D/2},\quad t>0\qquad \Longleftrightarrow \qquad \|\E^{-t\rh_0}\|_{\ell^1\to \ell^\infty} \le  C_2t^{-D/2},\quad  t>0.
\]
Here $C_1$ and $C_2$ are positive constants, which do not depend on $t$. Let us also mention that the estimates of this type are crucial in proving Rozenblum--Cwikel--Lieb (CLR) type estimates for both $\bH_\gA$ and $\rh_\gA$ (see Section \ref{ss:5.2}).

Our results continue and extend the previous work \cite{KM10a, KM10, KM_13} and \cite{kmn16} on 1-D Schr\"odinger operators and matrix Schr\"odinger operators with point interactions, respectively. Notice that (see Example \ref{ex:02}) in this case the line or a half-line can be considered as the simplest metric graph (a regular tree with $d=2$) and then the corresponding discrete Laplacian is simply a Jacobi (tri-diagonal) matrix (with matrix coefficients in the case of matrix Schr\"odinger operators).

Let us now finish the introduction by briefly describing the content of the article. The core of the paper is Section \ref{sec:BT_QG}, where we construct a suitable boundary triplet for the operator $\bH_{\max}$ (Theorem \ref{th:triplet} and Corollary \ref{cor:2.4}) by applying an efficient procedure suggested recently in \cite{KM10}, \cite{MalNei_08} (see also Appendix \ref{ss:a4}). 
  The central result of Section \ref{sec:BT_QG} is Theorem \ref{th:HTheta}, which describes basic spectral properties (self-adjointness, lower semiboundedness, spectral estimates, etc.) of proper extensions $\bH_{\Theta}$ of $\bH_{\min}$
 given by 
\be\label{eq:1.02}
\begin{split}
\bH_{\Theta}:= \bH_{\max}\upharpoonright &\dom(\bH_\Theta),\\
& \dom(\bH_\Theta):= \{f\in \dom(\bH_{\max})\colon \{\Gamma_0f,\Gamma_1\}\in\Theta\},
\end{split}
\ee
 in terms of the corresponding properties of {\em the boundary relation} $\Theta$. In particular, \eqref{eq:1.02} establishes a one-to-one correspondence between self-adjoint (respectively, symmetric) linear relations in an auxiliary Hilbert space $\cH$ and self-adjoint (respectively, symmetric) extensions of the minimal operator $\bH_{\min}$. 
 
In Section \ref{sec:parameterization} we specify Theorem \ref{th:HTheta} to the case of the Hamiltonian $\bH_{\gA}$. First of all, we find the boundary relation parameterizing the operator $\bH_{\gA}$ in the sense of \eqref{eq:1.02}. As it was already mentioned, its operator part is unitarily equivalent to the discrete Laplacian \eqref{eq:1.03}--\eqref{eq:1.04} and hence this fact establishes a close connection between spectral properties of $\bH_{\gA}$ and $\rh_\gA$ (Theorem \ref{th:main}).

In Sections \ref{sec:kirchhoff} and \ref{sec:delta}, we exploit recent advances in spectral theory of unbounded discrete Laplacians and prove a number of results on quantum graphs with Kirchhoff and $\delta$-couplings at vertices avoiding the standard restriction $\inf_{e\in\cE}|e|>0$. More specifically, the case of Kirchhoff conditions is considered in Section \ref{sec:kirchhoff}, where we prove several self-adjointness results and also provide estimates on the bottom of the spectrum as well as on the essential spectrum of $\bH_{0}$.
We discuss the self-adjointness of $\bH_{\gA}$ in Section
\ref{ss:5.1}. On the one hand, we show that $\bH_{\gA}$ is
self-adjoint for any $\gA\colon\cV\to \R$ whenever the weighted
degree function $\Deg$ is bounded on
$\cV$. In the case when $\Deg$ is locally bounded on $\cV$, we prove
self-adjointness and lower semiboundedness of $\bH_{\gA}$ under
certain semiboundedness assumptions on $\gA$. We also demonstrate by
simple examples that these results are sharp. Section \ref{ss:5.2}
is devoted to CLR-type estimates for quantum graphs. In Section
\ref{ss:5.3} we investigate spectral types of $\bH_\gA$. Moreover,
using the Cheeger-type estimates for $\rh_\alpha$, we prove several
spectral bounds for $\bH_\gA$.

As it was already mentioned, Theorem \ref{th:HTheta} is valid for
all self-adjoint extensions of $\bH_{\min}$, however, the
corresponding boundary relation may have a complicated structure
when we go beyond the $\delta$ couplings. In
Section \ref{sec:final}, we briefly discuss the case of
the so-called $\delta_s'$-couplings, cf. \cite{ex96} . It turns our
that the corresponding boundary operator is a difference operator,
however, its order depends on the vertex degree function of the underlying
discrete graph.

In Appendix \ref{app:triplets} we collect necessary definitions and facts about linear relations in Hilbert spaces, boundary triplets and Weyl functions.

\subsection*{Notation}  $\N$, $\Z$, $\R$, $\C$ have standard meaning; $\Z_{\ge 0} = \Z\cap[0,\infty)$.

 $a\vee b=\max(a,b)$, $a\wedge b = \min(a,b)$.

$\cH$ and $\gH$ denote separable complex Hilbert spaces; $\rI_{\gH}$ and $\bO_{\gH}$ are, respectively, the identity and the zero maps on $\gH$; $\rI_n:=\rI_{\C^n}$ and $\bO_n:=\bO_{\C^n}$. By $\cB(\gH)$ and $\cC(\gH)$ we denote, respectively, the sets of bounded and closed linear operators in $\gH$; $\wt{\cC}(\gH)$ is the set of closed linear relations in $\gH$; $\gS_p(\gH)$ is the two-sided von Neumann--Schatten ideal in $\gH$, $p\in(0,\infty]$. In particular, $\gS_1(\gH)$,  $\gS_2(\gH)$ and $\gS_\infty(\gH)$ denote the trace ideal, the Hilbert--Schmidt ideal and the set of compact operators in $\gH$. 

Let $T=T^\ast$ be a self-adjoint linear operator (relation) in $\gH$. For a Borel set $\Omega\subseteq \R$, by $E_T(\Omega)$ we denote the corresponding orthogonal spectral projection of $T$. Moreover, we set 
\begin{align*}
T^- & := T\,E_T\big({(-\infty,0)}\big), & \kappa_-(T) & = \dim \ran(T^-) = {\rm tr} (T^-),
\end{align*}
that is, $\kappa_-(T)$ is the total multiplicity of the negative spectrum of $T$. Note that $\kappa_-(T)$ is the number (counting multiplicities) of negative eigenvalues of $T$ if the negative spectrum of $T$ is discrete.
In this case we denote by $\lambda_j(T):=\lambda_j(|T^-|)$ their absolute values numbered in the decreasing order and counting their multiplicities.

\section{Boundary triplets for graphs}\label{sec:BT_QG}

Let us set up the framework. Let $\cG_d = (\cV,\cE)$ be {\em a (undirected) graph}, that is, $\cV$ is a finite or countably infinite set of vertices and $\cE$ is a finite or countably infinite set of edges. For two vertices $u$, $v\in \mathcal{V}$ we shall write $u\sim v$ if there is an edge $e_{u,v}\in \mathcal{E}$ connecting $u$ with $v$. 
 For every $v\in \mathcal{V}$, we denote  the set of edges incident to the vertex $v$ by $\cE_v$ and 
\be\label{eq:combdeg}
\deg(v):= \#\{e\colon e\in\cE_v\} 
\ee
is called {\em the degree} (or {\em combinatorial degree}) of a vertex $v\in\cV$. 
  {\em A path} $\cP$ of (combinatorial) length $n\in\Z_{\ge 0}$ is a subset of vertices $\{v_0,v_1,\dots, v_n\}\subset \cV$ such that $n$ vertices $\{v_0,v_1,\dots, v_{n-1}\}$ are distinct and $v_{k-1}\sim v_k$ for all $k\in \{1,\dots,n\}$. A graph $\cG_d$ is called {\em connected} if for any two vertices  there is a path  connecting them. 
  
We also need the following assumptions on the geometry of $\cG_d$:

\begin{hypothesis}\label{hyp:graph01}
$\cG_d$ is connected and there are no loops or multiple edges.
\end{hypothesis}

\begin{remark}\label{rem:assumptions}
Let us stress that the above assumptions can be made without loss of generality. Namely, if $\cG_d$ is not connected, then one simply needs to consider each connected component separately. The simplicity assumption can  always be achieved by adding the so-called {\em inessential vertices} (vertices of degree two and equipped with Kirchhoff conditions) to the corresponding metric graph. Indeed, adding or removing such a vertex does not change spectral properties of the corresponding quantum graph (see, e.g., \cite[Remark 1.3.3]{bk13}).
\end{remark}

Let us assign each edge $e\in\cE$ with length $|e|\in (0,\infty)$\footnote{We shall always assume that there are no edges having an infinite length, however, see Remark \ref{rem:3.01}(ii).} and direction\footnote{This means that the graph $\cG_d$ is directed}, that is, each edge $e\in \cE$ has one initial vertex $e_o$ and one terminal vertex $e_i$. In this case $\cG=(\cV,\cE,|\cdot|) = (\cG_d,|\cdot|)$ is called {\em a metric graph}. Moreover, every edge $e\in\cE$ can be identified with the interval $(0,|e|)$ and hence we can introduce the Hilbert space $L^2(\cG)$ of functions $f\colon \cG\to \C$ such that 
\[
L^2(\cG) = \bigoplus_{e\in\cE} L^2(e) = \Big\{f=\{f_e\}_{e\in\cE}\colon f_e\in L^2(e),\ \sum_{e\in\cE}\|f_e\|^2_{L^2(e)}<\infty\Big\}.
\] 
 
Let us equip  $\cG$ with the Laplace operator. For every $e\in\cE$ consider the maximal operator $\rH_{e,\max}$ acting on functions $f\in W^{2,2}(e)$ as a negative second derivative. Now consider the maximal operator on $\cG$ defined by
\begin{align}\label{eq:Hmax}
\bH_{\max} & = \bigoplus_{e\in \cE} \rH_{e,\max}, & \rH_{e,\max} & = -\frac{\rD^2}{\rD x_e^2},\quad \dom(\rH_{e,\max}) = W^{2,2}(e).
\end{align}
 For every $f_e\in W^{2,2}(e)$ the following quantities 
 \begin{align}\label{eq:tr_fe}
 f_e(e_o) & := \lim_{x\to e_o} f_e(x), & f_e(e_i) & := \lim_{x\to e_i} f_e(x),
 \end{align}
 and 
 \begin{align}\label{eq:tr_fe'}
 f_e'(e_o) & := \lim_{x\to e_o} \frac{f_e(x) - f_e(e_o)}{|x - e_o|}, & f_e'(e_i) & := \lim_{x\to e_i} \frac{f_e(x) - f_e(e_i)}{|x - e_i|},
 \end{align}
 are well defined. 

We begin with a simple and well known fact (see, e.g., \cite{KM10}).

\begin{lemma}\label{lem:triplet_e}
Let $e\in\cE$ and $\rH_{e,\max}$ be the corresponding maximal operator. The triplet $\Pi_e^0 = \{\C^2,{\Gamma}^0_{0,e},{\Gamma}^0_{1,e}\}$, where the mappings ${\Gamma}^0_{0,e}$, ${\Gamma}^0_{1,e}\colon W^{2,2}(e) \to \C^2$ are defined by
\begin{align}\label{eq:G00}
{\Gamma}^0_{0,e} & \colon f\mapsto \begin{pmatrix} f_e(e_o) \\[2mm] f_e(e_i) \end{pmatrix}, & \Gamma^0_{1,e} & \colon f\mapsto \begin{pmatrix}  f_e'(e_o) \\[2mm]  f_e'(e_i) \end{pmatrix},
\end{align}
is a boundary triplet for $\rH_{e,\max}$. Moreover, the corresponding Weyl function ${M}^0_e \colon \C\setminus\R\to \C^{2\times 2}$ is given by\footnote{Here $\sqrt{z}$ denotes the principal branch of the square root with the cut along the negative semi-axis.}
\be\label{eq:Me0}
{M}^0_e\colon z\mapsto  \sqrt{z}\begin{pmatrix} -\cot(|e|\sqrt{z}) & {\rm csc}(|e|\sqrt{z}) \\ {\rm csc}(|e|\sqrt{z}) & -\cot(|e|\sqrt{z}) \end{pmatrix}.
\ee 
\end{lemma}

\begin{proof}
The proof is straightforward and we leave it to the reader. 
\end{proof}

It is easy to see that the Green's formula
\begin{align}\label{eq:Green}
\begin{split}
(\bH_{\max}f,g)_{L^2(\cG)} - (f,\bH_{\max}g)_{L^2(\cG)} 
&= \sum_{e\in\cE}  f_e'(e_i)g_e(e_i)^* - f_e(e_i)(g_e'(e_i))^* \\
 & \quad + \sum_{e\in\cE}  f_e'(e_o)g_e(e_o)^* - f_e(e_o)(g_e'(e_o))^*\\
&= \sum_{v\in\cV} \sum_{e\in\cE_v} f_e'(v)g_e(v)^* - f_e(v) (g_e'(v))^*,
\end{split} 
\end{align}
holds for all $f$, $g\in \dom(\bH_{\max})\cap L^2_c(\cG)$, where $L^2_c(\cG)$ is a subspace consisting of functions from $L^2(\cG)$ vanishing everywhere on $\cG$ except finitely many edges, and the asterisk denotes complex conjugation. 
 One would expect that a boundary triplet for $\bH_{\max}$ can be constructed as a direct sum $\Pi = \oplus_{e\in\cE} \Pi_e^0$ of boundary triplets $\Pi_e^0$, however, it is not true once $\inf_{e\in\cE} |e| = 0$ (see \cite{KM10} for further details). Using Theorem \ref{th_bt_criterion}, we proceed as follows (see also \cite[Section 4]{KM10}). 
First of all, \eqref{eq:Me0} extends to a meromorphic function with simple poles $\frac{\pi^2}{|e|^2}k^2$, $k\in\N$. Hence for every $e\in\cE$ we set 
\begin{align}\label{eq:RQ_e}
\rR_e & := \sqrt{|e|}I_2, & \rQ_e & := \lim_{z\to 0}{M}^0_e(z) = \frac{1}{|e|}\begin{pmatrix} -1 & 1 \\ 1 & -1 \end{pmatrix},
\end{align}
and then we define the new mappings ${\Gamma}_{0,e}$, ${\Gamma}_{1,e}\colon W^{2,2}(e) \to \C^2$ by 
\begin{align}\label{eq:Ge}
\Gamma_{0,e} & := \rR_e {\Gamma}^0_{0,e}, & \Gamma_{1,e} & := \rR_e^{-1}( {\Gamma}^0_{1,e} - \rQ_e{\Gamma}^0_{0,e}),
\end{align}
that is, 
\begin{align}\label{eq:G01_e}
 \Gamma_{0,e} & \colon f\mapsto \begin{pmatrix} \sqrt{|e|}f_e(e_o) \\[2mm]  \sqrt{|e|}f_e(e_i) \end{pmatrix}, & \Gamma_{1,e} & \colon f\mapsto 
 \frac{1}{|e|^{3/2}}\begin{pmatrix} |e|f_e'(e_o) +f_e(e_o) - f_e(e_i) \\[2mm]  |e|f_e'(e_i) - f_e(e_o) + f_e(e_i)\end{pmatrix}.
\end{align}
Clearly, $\Pi_e = \{\C^2,\Gamma_{0,e},\Gamma_{1,e}\}$ is also a boundary triplet for $\rH_{\max,e}$. In addition, the following claim holds.

\begin{theorem}\label{th:triplet}
Suppose $\sup_{e\in\cE} |e|<\infty$. Then the direct sum of boundary triplets
\be
\Pi = \bigoplus_{e\in\cE} \Pi_e = \{\cH,\Gamma_0,\Gamma_1\},\qquad \cH=\bigoplus_{e\in\cE}\C^2,\quad \Gamma_j:=\bigoplus_{e\in\cE} \Gamma_{j,e},\ \ j\in\{0,1\},
\ee
is a boundary triplet for the operator $\bH_{\max}$. Moreover, the corresponding Weyl function is given by
\begin{align}\label{eq:M}
M(z) & = \bigoplus_{e\in\cE} M_e(z), & 
 M_e(z) & = \rR_e^{-1}({M}^0_e(z) - \rQ_e)\rR_e^{-1},
\end{align}
for all $z\in\C\setminus\R$.
\end{theorem}

\begin{proof}
By Theorem \ref{th_bt_criterion}, we need to verify either of the conditions \eqref{WF_criterion} or \eqref{III.2.2_02}. However, this can be done as in the proof of \cite[Theorem 4.1]{KM10} line by line since 
\[
M_e(z) = \frac{\sqrt{z}}{|e|}\begin{pmatrix} -\cot(|e|\sqrt{z}) + \frac{1}{|e|} & {\rm csc}(|e|\sqrt{z}) - \frac{1}{|e|} \\ {\rm csc}(|e|\sqrt{z}) - \frac{1}{|e|} & -\cot(|e|\sqrt{z}) + \frac{1}{|e|} \end{pmatrix},\quad z\in\C\setminus\R,
\]
and we omit the details.
\end{proof}

Moreover, similarly to \cite[Proposition 4.4]{KM10} one can also prove the following

\begin{lemma}\label{lem:Munif}
Suppose $\sup_{e\in\cE} |e|<\infty$. Then the Weyl function $M(x)$ given by \eqref{eq:M} uniformly 
tends to $-\infty$ as $x\to -\infty$, that is, for every $N>0$ there is $x_N<0$ such that
\[
M(x) <-N\cdot{\rm I}_{\cH}
\]
for all $x<x_N$.
\end{lemma}

We shall also need another boundary triplet for $\bH_{\max}$, which can be obtained from the triplet $\Pi$ by regrouping all its components with respect to the vertices:
\be\label{eq:tiPi}
\cH_{\cG} = \bigoplus_{v\in\cV}\C^{\deg(v)},\qquad \wt{\Gamma}_j = \bigoplus_{v\in\cV}\wt\Gamma_{j,v},\quad j\in\{0,1\},
\ee
where
\be\label{eq:tiG0}
\wt\Gamma_{0,v}\colon f\mapsto \big\{\sqrt{|e|} f_e(v)\big\}_{e\in\cE_v}, 
\ee
and 
\be\label{eq:tiG1}
\wt\Gamma_{1,v}\colon f\mapsto \big\{|e|^{-1/2} f'_e(v) + (-1)^{q_e(v)}|e|^{-3/2}(f_e(e_o) - f_e(e_i))\big\}_{e\in\cE_v}, 
\ee
with 
\be\label{eq:qpm}
q_e(v):=\begin{cases} 1, & v=e_o,\\ -1, & v=e_i. \end{cases}
\ee

\begin{corollary}\label{cor:2.4}
If $\sup_{e\in\cE} |e|<\infty$, then the triplet $\Pi_{\cG} = \{\cH_{\cG}, \wt\Gamma_0,\wt\Gamma_1\}$ given by \eqref{eq:tiPi}--\eqref{eq:qpm} is a boundary triplet for $\bH_{\max}$.
\end{corollary}

\begin{proof}
Every $f\in \cH$ and $\ti{f}\in \cH_{\cG}$ can be written as follows $f = \{(f_{eo},f_{ei})\}_{e\in\cE}$ and 
$\ti{f} = \{(\ti{f}_{e,v})_{e\in\cE_v}\}_{v\in\cV}$, respectively. Define the operator $U:\cH\to \cH_{\cG}$ by 
\be\label{eq:U}
U \colon  \{(f_{eo},f_{ei})\}_{e\in\cE} \mapsto \{({f}_{e,v})_{e\in\cE_v}\}_{v\in\cV},\quad f_{e,v}:=\begin{cases}f_{eo}, & v=e_o,\\ f_{ei}, & v=e_i.\end{cases}
\ee
Clearly, $U$ is a unitary operator and moreover
\be\label{eq:tiGamma}
\wt\Gamma_j = U\Gamma_j,\quad j\in\{0,1\}.
\ee
This completes the proof.
\end{proof}

Let us also mention other important relations. 

\begin{corollary}\label{cor:2.05}
The Weyl function $M_\cG$ corresponding to the boundary triplet  \eqref{eq:tiPi}--\eqref{eq:qpm} is given by
\be\label{eq:tiM}
M_\cG(z) = UM(z)U^{-1},
\ee
where $M$ is the Weyl function corresponding to the triplet $\Pi$ constructed in Theorem \ref{th:triplet} and $U$ is the operator defined by \eqref{eq:U}. 
\end{corollary}

\begin{remark}
If
\begin{align*}
\Gamma_0^0 & :=  \bigoplus_{e\in\cE}\wt\Gamma_{j,e}^0, & \Gamma_1^0 & :=  \bigoplus_{e\in\cE}\wt\Gamma_{1,e}^0,
\end{align*}
where  $\Gamma_0^0$ and $\Gamma_1^0$ are given by \eqref{eq:G00}, then 
\begin{align}\label{eq:G2.19}
\wt\Gamma_0^0 & := U\Gamma_0^0, & \wt\Gamma_1^0 & := U\Gamma_1^0,
\end{align}
have the following form
\be\label{eq:wtG0}
\wt\Gamma_0^0 = \bigoplus_{v\in\cV}\wt\Gamma_{0,v}^0,\qquad \wt\Gamma_{0,v}^0\colon f\mapsto  \{f_e(v)\}_{e\in\cE_v},
\ee
and 
\be\label{eq:wtG1}
\wt\Gamma_1^0 = \bigoplus_{v\in\cV}\wt\Gamma_{1,v}^0,\qquad \wt\Gamma_{1,v}^0\colon f\mapsto  \{f'_e(v)\}_{e\in\cE_v}.
\ee
\end{remark}

\begin{corollary}\label{cor:Munif}
Let $M_\cG$ be the Weyl function corresponding to the boundary triplet $\Pi_\cG$. Then $M_\cG(x)$ uniformly tends to $-\infty$ as $x\to -\infty$.
\end{corollary}

\begin{proof}
It is an immediate consequence of Lemma \ref{lem:Munif} and \eqref{eq:tiM}. 
\end{proof}

Let $\Theta$ be a linear relation in $\cH_{\cG}$ and define the following operator
\be\label{eq:Htheta}
\begin{split}
\bH_{\Theta}:=\bH_{\max}\upharpoonright &{\dom(\bH_{\Theta})},\quad \\
&\dom(\bH_{\Theta}):=\big\{f\in \dom(\bH_{\max})\colon \{\wt\Gamma_0f,\wt\Gamma_1f\}\in\Theta\big\},
\end{split}
\ee
where the mappings $\wt\Gamma_0$ and $\wt\Gamma_1$ are defined by \eqref{eq:tiPi}--\eqref{eq:tiG1}. Since $\Pi_\cG$ is a boundary triplet for $\bH_{\max}$, every proper extension of the operator $\bH_{\min}$ has the form \eqref{eq:Htheta}. Moreover, by Theorem \ref{th:properext},  \eqref{eq:Htheta} establishes a bijective correspondence between the set $\Ext(\bH_{\min})$ of proper extensions of $\bH_{\min}$ and the set of all linear relations in $\cH_{\cG}$. The next result summarizes basic spectral properties of operators $\bH_{\Theta}$ characterized in terms of the corresponding boundary relation $\Theta$. In particular, we are able to describe all self-adjoint extensions of the minimal operator $\bH_{\min}$. 

\begin{theorem}\label{th:HTheta}
Suppose $\sup_{e\in\cE} |e|<\infty$. Also, let $\Theta$ be a linear relation in $\cH_{\cG}$ and let $\bH_{\Theta}$ be the corresponding operator \eqref{eq:Htheta}. Then:
\begin{itemize}
\item[(i)] $\bH_{\Theta}^\ast = \bH_{\Theta^\ast}$.
\item[(ii)] $\bH_{\Theta}$ is closed if and only if the linear relation $\Theta$ is closed.
\item[(iii)] $\bH_{\Theta}$ is symmetric if and only if $\Theta$ is symmetric and, moreover,
\[
{\rm n}_\pm(\bH_{\Theta}) = {\rm n}_\pm(\Theta).
\] 
In particular, $\bH_{\Theta}$ is self-adjoint if and only if so is  $\Theta$.
\end{itemize}  

Assume in addition that $\Theta$ is a self-adjoint linear relation (hence $\bH_{\Theta}$ is also self-adjoint). Then:
\begin{itemize}
\item[(iv)] $\bH_{\Theta}$ is lower semibounded if and only if the same is true for $\Theta$.
\item[(v)]  $\bH_{\Theta}$ is nonnegative (positive definite) if and only if $\Theta$ is nonnegative (positive definite).
\item[(vi)] The total multiplicities of negative spectra of $\bH_{\Theta}$ and $\Theta$ coincide,
\be\label{eq:kappaT}
\kappa_-(\bH_{\Theta}) = \kappa_-(\Theta).
\ee
\item[(vii)] For every $p\in(0,\infty]$ the following equivalence holds
\be\label{eq:SpnegTheta}
\bH_{\Theta}^- \in \mathfrak{S}_p(L^2(\cG)) \ \Longleftrightarrow \ \Theta^- \in \mathfrak{S}_p(\cH_{\cG}).
\ee
\item[(viii)] If the negative spectrum of $\bH_{\Theta}$ (or equivalently $\Theta$) is discrete, then for every $\gamma\in (0,\infty)$ the following equivalence holds
\be\label{eq:weakSp}
\lambda_j(\bH_{\Theta}) =  j^{-\gamma}(a+o(1)) \ \Longleftrightarrow \ \lambda_j({\Theta})=  j^{-\gamma}(b+o(1)),
\ee
as $j\to \infty$, where either $ab\not = 0$ or $a = b = 0$.

\item[(ix)] If, in addition, $\Theta$ is lower semibounded, then $\inf\sigma_{\ess}(\bH_\Theta)>0$ $(\inf\sigma_{\ess}(\bH_\Theta)=0)$ holds exactly when $\inf\sigma_{\ess}(\Theta)>0\ ($respectively, $\inf\sigma_{\ess}(\Theta)=0)$.

\item[(x)] Also, let $\wt\Theta=\wt\Theta^\ast \in \widetilde{\mathcal{C}}(\cH_{\cG})$. Then for every $p\in(0,\infty]$ the following equivalence holds for the corresponding Neumann--Schatten ideals
\be\label{eq:SpequivT}
          (\bH_{\Theta}-\I)^{-1} - (\bH_{\wt{\Theta}}-\I)^{-1}\in{\mathfrak S}_p(L^2(\cG)) \ 
\Longleftrightarrow \          
                  (\Theta-\I)^{-1} - (\widetilde{\Theta}-\I)^{-1}\in{\mathfrak S}_p(\cH_{\cG}).
\ee
      If $\dom(\Theta) = \dom(\widetilde{\Theta})$ holds in addition, then 
\be\label{eq:SpimplicT}
         \overline{\Theta - \widetilde{\Theta}} \in \mathfrak{S}_p(\cH_{\cG})\ \Longrightarrow \ 
         ({\bH}_{\Theta}-\I)^{-1} - ({\bH}_{\widetilde{\Theta}}-\I )^{-1}\in {\mathfrak S}_p(L^2(\cG)).
\ee
\item[(xi)] The spectrum of $\bH_\Theta$ is purely discrete if and only if $\#\{e\in\cE\colon |e|>\varepsilon\}$ is finite for every $\varepsilon>0$ and the spectrum of the linear relation $\Theta$ is purely discrete.
\end{itemize} 
\end{theorem}

\begin{proof}
Consider the boundary triplet $\Pi$ constructed in Theorem \ref{th:triplet}. Items (i), (ii), (iii) and (x) follow from Theorem \ref{th:properext}. Item (iv) follows from Theorem \ref{th:lsb=lsb} and Corollary \ref{cor:Munif}.

Next consider the corresponding Weyl function $M$ given by \eqref{eq:M}. Clearly, 
\[
M_e(0) = \rR_e^{-1}(M_e^0(0) - \rQ_e)\rR_e^{-1} = \rR_e^{-1}(\rQ_e - \rQ_e)\rR_e^{-1} = \bO_{2}
\]
for all $e\in\cE$. Then \eqref{eq:M} together with \eqref{eq:srMa} implies that $M(0) = \bO_{\cH} \in \cB(\cH)$. Moreover, in view of \eqref{eq:tiM}, we get
\[
M_\cG(0) = UM(0)U^{-1} = \bO_{\cH_{\cG}} \in [\cH_{\cG}].
\]
Noting that 
\[
\rH_e^0:= \rH_{e,\max} \upharpoonright {\ker(\Gamma_{0,e})} = \rH_e^F
\]
is the Friedrichs extension of $\rH_{e,\min} =( \rH_{e,\max})^\ast$, we immediately conclude that 
\be\label{eq:H0}
\bH^0:= \bH_{\max}\upharpoonright {\ker(\wt\Gamma_{0})} = \bH_{\max}\upharpoonright {\ker(\Gamma_{0})} = \bigoplus_{e\in\cE} \rH_e^0 = \bH^F
\ee
is the Friedrichs extension of $\bH_{\min} = (\bH_{\max})^\ast$.  
Moreover, \be\label{eq:sigmaHe}
\sigma(\rH_e^0) = \Big\{ \frac{\pi^2 n^2}{|e|^2}\Big\}_{n\in\N}
\ee
and hence
\[
\inf \sigma(\bH^F) =\inf_{e\in\cE} \inf\sigma(\rH_e^F)= \inf_{e\in\cE} \frac{\pi^2}{|e|^2} = \frac{\pi^2}{(\sup_{e\in\cE} |e|)^2}>0.
\]
Now items (v)--(viii) follow from Theorem \ref{th:lsbTheta} and item (ix) follows from Theorem \ref{th:ess-impl-ess}.

Finally, it follows from \eqref{eq:H0} and \eqref{eq:sigmaHe} 
that the spectrum of $\bH^F$ is purely discrete if and only if $\#\{e\in\cE\colon |e|>\varepsilon\}$ is finite for every $\varepsilon>0$. This fact together with Theorem \ref{th:properext}(iv) implies item (xi). 
\end{proof}

\begin{remark}
The assumption $\sup_{e\in\cE} |e|<\infty$ in Theorem \ref{th:HTheta} can be dropped either by modifying the underlying graph $\cG$ by adding additional vertices or by modifying the construction of the boundary triplet in Theorem \ref{th:triplet}. However, both options lead to a more cumbersome form of the corresponding boundary relation $\Theta$ and we decided to exclude this case from our considerations in order to keep the exposition as transparent as possible. 
\end{remark}

\begin{remark}
The analogs of statements (iii)  and (iv) of Theorem \ref{th:HTheta} were obtained in \cite{lsv14} under the additional restrictive assumption $\inf_{e\in\cE} |e|>0$. Notice that if the latter holds, then the regularization \eqref{eq:RQ_e}--\eqref{eq:G01_e} is not needed and one can construct a boundary triplet for the maximal operator $\bH_{\max}$ by summing up the triplets \eqref{eq:G00}.
\end{remark}

\section{Parameterization of quantum graphs with $\delta$-couplings}\label{sec:parameterization}

Turning to a more specific problem, we need to make further assumptions on the geometry of a connected metric graph $\cG$.

\begin{hypothesis}\label{hyp:graph02}
 $\cG$ is locally finite, that is, every vertex $v\in\cV$ has finitely many neighbours,  $1\le \deg(v)<\infty$ for all $v\in \cV$. Moreover, there is a finite upper bound on the lengths of edges, 
 \be\label{eq:dstar}
 \sup_{e\in\cE}|e|<\infty.
\ee
\end{hypothesis}

Let $\alpha\colon \cV\to \R$ be given and equip every vertex $v\in\cV$ with the so-called $\delta$-type vertex condition: 
\be\label{eq:bcalpha}
\begin{cases} f\ \text{is continuous at}\ v,\\[2mm] \sum_{e\in \cE_v}f_e'(v) = \alpha(v)f(v), \end{cases}
\ee
Let us define the operator $\bH_{\alpha}$ as the closure of the operator $\bH_{\alpha}^0$ given by
\be\label{eq:Halpha}
\begin{split}
\bH_{\alpha}^0 = &\bH_{\max}\upharpoonright {\dom(\bH_{\alpha}^0)},\\ 
&\dom(\bH_{\alpha}^0) = \{f\in \dom(\bH_{\max})\cap L^2_{c}(\cG)\colon f\ \text{satisfies}\ \eqref{eq:bcalpha},\ v\in\cV\},
\end{split}
\ee
where $L^2_c(\cG)$ consists of functions from $L^2(\cG)$ vanishing everywhere on $\cG$ except finitely many edges.

\begin{remark} \label{rem:3.01}
A few remarks are in order:
\begin{itemize}
\item[(i)] If $\deg(v_0)=\infty$ for some $v_0\in\cV$, then it was shown in \cite[Theorem 5.2]{lsv14} that a Kirchhoff-type boundary condition at $v_0$ (as well as \eqref{eq:bcalpha}) leads to an operator which is not closed. Moreover, it turns out that its closure gives rise to Dirichlet boundary condition at $v_0$, i.e.,
disconnected edges.

\item[(ii)] Assumption \eqref{eq:dstar} is of a technical character. Of course, the case of edges having an infinite length would require separate considerations in Section \ref{sec:BT_QG} and this will be done elsewhere. On the other hand, the case when all edges have finite length but there is no uniform upper bound can be reduced to the case of graphs satisfying \eqref{eq:dstar} either by adding additional inessential vertices or by slight modifications in the considerations of Section \ref{sec:BT_QG}. Note also that those allow to include situations when the graph is not simple, that is, it has loops and multiple edges (cf. Hypothesis \ref{hyp:graph01}). 
\end{itemize}
\end{remark}

Let us emphasize that the operator  $\bH_\alpha$ is symmetric. Moreover, simple examples show that $\bH_\alpha$ might not be self-adjoint. 

  \begin{example}[1-D Schr\"odinger operator with $\delta$-interactions]\label{ex:01}
Consider the positive semi-axis $(0,\infty)$ and let $\{x_k\}_{k\ge 0}\subset [0,\infty)$ be a strictly increasing sequence such that $x_0=0$ and $x_k\uparrow +\infty$. Considering  $x_k$
as vertices and the intervals $e_k=(x_{k-1},x_k)$ as edges, we end up with the simplest infinite metric graph. Notice that for every real sequence  $\alpha=\{\alpha_k\}_{k\ge 0}$ with $\alpha_0=0$ conditions
\eqref{eq:bcalpha} take the following form: $f'(0)=0$  and
   \begin{equation}\label{eq:delta_line}
   \begin{split}
&f(x_k-) = f(x_k+)=:f(x_k), \\
 &f'(x_k + ) - f'(x_k - ) = \alpha_k f(x_k),\quad  k\in\N.
\end{split}
   \end{equation}
The operator $\bH_\alpha$ is known as {\em the one-dimensional Schr\"odinger operator with $\delta$-interactions} on  $X=\{x_k\}_{k\in\N}$ (see, e.g., \cite{AGHH}), and the corresponding differential expression is given by
  \be\label{eq:1Ddelta}
\rH_{X,\gA} = -\frac{\rD^2}{\rD x^2} + \sum_{k\in\N}
\alpha_k\delta(x-x_k).
  \ee

It was proved in  \cite{KM10} that $\rH_{X,\gA}$ is self-adjoint if  
$\sum_{k}|e_k|^2= \infty$ (the latter is known in the literature as {\em the Ismagilov condition}, see \cite{is}). 
On the other hand (see  \cite[Proposition 5.9]{KM10}), if $\sum_{k}|e_k|^2<\infty$ and in addition $|e_{k-1}| \cdot |e_{k + 1}|\geq
|e_k|^2$ for all sufficiently large $k$, then the operator $\rH_{X,\gA}$ is symmetric with $\mathrm{n}_\pm(\rH_{X,\gA})=1$ whenever  $\alpha=\{\alpha_k\}_{k\in\N}$ satisfies the following condition
  \[
  \sum_{k=1}^\infty |e_{k + 1}|\Big|\alpha_k + \frac{1}{|e_k|} + \frac{1}{|e_{k +1}|}\Big|<\infty.
  \]
This effect was discovered by C.\ Shubin Christ and G.\ Stolz \cite[pp. 495--496]{ShuSto94} in the special case  $|e_k| = 1/k$
and $\alpha_k =-(2k+1)$, $k\in\N$. For further details and results we refer to \cite{KM_13}, \cite{mir15}. \hfill$\lozenge$
\end{example}

 Our main aim is to find a boundary relation $\Theta_\gA$ parameterizing the operator $\bH_\alpha$ in terms of the boundary triplet $\Pi_\cG$ given by \eqref{eq:tiPi}--\eqref{eq:tiG1}. First of all, notice that at each vertex $v\in \cV$ the boundary conditions \eqref{eq:bcalpha} have the following form
\be\label{eq:3.6}
D_{v} \wt\Gamma_{1,v}^0f = C_{v,\alpha} \wt\Gamma_{0,v}^0f , 
\ee
where $\wt\Gamma_{0,v}^0f = \{f_e(v)\}_{e\in\cE_v}$, $\wt\Gamma_{1,v}^0f = \{f'_e(v)\}_{e\in\cE_v}$ (see \eqref{eq:wtG0} and \eqref{eq:wtG1}) and the matrices $C_{v,\alpha}$, $D_{v}\in \C^{\deg(v)\times\deg(v)}$ are given by
\be\label{eq:3.7}
\begin{split}
C_{v,\alpha} & = \begin{pmatrix} 
1 & -1 & 0 & 0 & \dots & 0 \\
0& 1 & -1 & 0 &  \dots & 0 \\
0& 0 &1 & -1 &  \dots & 0 \\
\dots &\dots &\dots &\dots &\dots &\dots \\ 
0 & 0 & 0 & 0 & \dots & -1 \\
\alpha(v) & 0 & 0 & 0 & \dots & 0
\end{pmatrix},\\[2mm]
D_v &= \begin{pmatrix} 
0 & 0 & 0 & 0 & \dots & 0 \\
0& 0 & 0 & 0 &  \dots & 0 \\
0& 0 &0 & 0 &  \dots & 0 \\
\dots &\dots &\dots &\dots &\dots &\dots \\ 
0 & 0 & 0 & 0 & \dots & 0 \\
1 & 1 & 1 & 1 & \dots & 1
\end{pmatrix}.
\end{split}
\ee
 It is easy to check that these matrices satisfy the Rofe--Beketov conditions  (see Proposition \ref{prop:RB}), that is
\begin{align}\label{eq:Rofe}
C_{v,\alpha}D_v^\ast & = D_vC_{v,\alpha}^\ast, & {\rm rank}(C_{v,\alpha}| D_v) & = \deg(v),
\end{align}
and hence 
\[
\Theta_{v,\alpha}:= \big\{\{f,g\} \in \C^{\deg(v)}\times\C^{\deg(v)}\colon\ C_{v,\alpha}f=D_v g\big\}
\]
is a self-adjoint linear relation in $\C^{\deg(v)}$.  
Now set
\begin{align*}
C_\alpha^0 & := \bigoplus_{v\in\cV} C_{v,\alpha}, & D^0 & := \bigoplus_{v\in\cV} D_{v}.
\end{align*}
Clearly, $f\in\dom(\bH_{\max})\cap L^2_c(\cG)$ satisfies
\[
D^0\wt\Gamma_1^0 f = C_\alpha^0 \wt\Gamma_0^0f,
\]
if and only if $f\in \dom(\bH_\alpha^0)=\dom(\bH_{\gA})\cap L^2_c(\cG)$. 
 In view of \eqref{eq:G2.19}, we get 
\begin{align*}
\wt\Gamma_0 & = \wt{\rR} \wt\Gamma_0^0, & \wt\Gamma_1 & = \wt{\rR}^{-1} (\wt\Gamma_1^0 - \wt{\rQ} \wt\Gamma_0^0)
\end{align*}
where 
\begin{align*}
\wt{\rR} & = U\rR U^{-1}, & \wt{\rQ} & = U\rQ U^{-1},
\end{align*}
and $\rR = \oplus_{e\in\cE}\rR_e$, $\rQ= \oplus_{e\in\cE}\rQ_e$ and $U$ are defined by \eqref{eq:RQ_e} and \eqref{eq:U}, respectively. 
Hence we conclude that $f\in \dom(\bH_\alpha^0)$ if and only if $f$ satisfies
\[
D\wt\Gamma_1 f = C_\alpha \wt\Gamma_0 f,
\]
where
\begin{align*}
D & = D^0\wt{\rR}, & C_\alpha & = (C_\alpha^0 - D^0\wt{\rQ})\wt{\rR}^{-1}.
\end{align*}

Thus we are led to specification of the boundary relation parameterizing the operator $\bH_{\gA}^0$. Namely, consider now the linear relation $\Theta_\alpha^0$ defined in $\cH_{\cG}$ by
\be\label{eq:ThetaA}
\Theta_\alpha^0 = \{\{f,g\}\in\cH_{\cG,c}\times \cH_{\cG,c}\colon  \ C_\alpha f = Dg\},
\ee
where $\cH_{\cG,c}$ consists of vectors of $\cH_{\cG}$ having only finitely many nonzero coordinates. 
It is not difficult to see that $\Theta_\alpha^0$ is symmetric and hence it admits the decomposition (see Appendix \ref{ss:a1})
\begin{align*}
\Theta_\alpha^0 & = \Theta_{\rm op}^0\oplus \Theta_{\mul}^0, & \Theta_{\mul}^0 & = \{0\}\times \mul(\Theta_\alpha^0),
\end{align*}
and $\Theta_{\rm op}^0$ is the operator part of $\Theta_\alpha^0$. Clearly,
\[
\mul(\Theta_\alpha^0) = \ker(D)\cap \cH_{\cG,c} = {\wt{\rR}^{-1} \ker(D^0)}\cap \cH_{\cG,c}.  
\]

Let   $f=\{f_v\}_{v\in\cV}\in \cH_{\cG}$, where $f_v=\{f_{v,e}\}_{e\in\cE_v}$. Next we observe that  
\begin{align*}
\wt{\rR} & = \bigoplus_{v\in\cV}\wt{\rR}_v, & \wt{\rR}_v & = \diag(\sqrt{|e|})_{e\in\cE_v},
\end{align*}
and 
\begin{align*}
\wt{\rQ} & = \bigoplus_{v\in\cV}\wt{\rQ}_{v} + \wt{\rQ}^0, & \wt{\rQ}_{v} & = -\diag(|e|^{-1})_{e\in\cE_v}, 
\end{align*}
where
\[
 (\wt{\rQ}^0f)_{v,e} = |e_{v,u}|^{-1}f_{u,e},\qquad u := \begin{cases} e_i, & v=e_o,\\ e_o, & v=e_i.\end{cases}
\]
Noting that 
\[
\cH_{\rm op} = \overline{\dom(\Theta_\alpha^0)} = \overline{\ran(D^*)} = \overline{\ran(\wt{\rR}(D^0)^*)},
\]
we get
\begin{align*}
\cH_{\rm op} & = \Span\{\f_v\}_{v\in\cV}, & \f_v & = \{f_{u,e}^v\},\quad f_{u,e}^v = \begin{cases} \sqrt{|e|}, & u=v,\\ 0, & u\neq v.\end{cases} 
\end{align*}
Let us now show that $\f_v\in\dom(\Theta_\alpha^0)$ for every $v\in\cV$. Denote by $P_v$ the orthogonal projection in $\cH_{\cG}$ onto $\cH_{\cG}^v:=\Span\{\f_v\}$.  Next notice that
\[
P_uC_\alpha \f_v = P_u(C_\alpha^0 - D^0\wt{\rQ})\wt{\rR}^{-1}\f_v = \begin{cases} \big(\underbrace{0,0,\dots,0,\alpha(v) + \sum_{e\in\cE_v}|e|^{-1}}_{\deg(v)}\big),& u=v,\\[2mm]
\big(\underbrace{0,0,\dots,0,-|e_{u,v}|^{-1}}_{\deg(u)}\big),& u\sim v,\\[2mm]
0,& u\not\sim v,\ u\neq v.\end{cases} 
\]
Finally, take $g\in \cH_{\cG,c}$ and consider
\[
(Dg)_u = (D^0\wt{\rR}g)_u = \big(\underbrace{0,0,\dots,0,\sum_{e\in\cE_u}\sqrt{|e|}g_{u,e}}_{\deg(u)} \big).
\] 
Therefore, define $\g_v\in \cH_{\rm op}$ by
\be\label{eq:g=f}
P_u \g_v = \{\sqrt{|e|}\}_{e\in\cE_u} \times
\begin{cases} 
\frac{1}{m(v)}(\alpha(v) + \sum_{e\in\cE_v}|e|^{-1}), & u=v, \\[2mm]
-\frac{1}{\sqrt{|e_{u,v}|} m(u)}, & u\sim v,\\[2mm]
0, & u\not\sim v,\ u\neq v,
\end{cases}
\ee
where the function $m\colon\cV\to (0,\infty)$ is given by
\be\label{eq:m_weight}
m\colon v\mapsto \sum_{e\in\cE_v}|e|,\qquad v\in\cV.
\ee
Clearly,
\[
C_\alpha \f_v = D \g_v,
\]
and hence ${\bf f}_v\in \dom(\Theta_\alpha^0)$. Moreover, \eqref{eq:g=f} immediately implies that 
\be
\g_v  = \frac{1}{m(v)}\Big(\alpha(v) + \sum_{e\in\cE_v}|e|^{-1}\Big) \f_v - \sum_{u\sim v} \frac{1}{\sqrt{|e_{u,v}|} m(u)} \f_u =: \Theta_{\rm op}^0\f_v.
\ee
Noting that $\{\f_v\}_{v\in\cV}$ is an orthogonal basis in ${\cH}_{\rm op}$ and $\|\f_v\|^2 = m(v)$ for all $v\in\cV$, we conclude that the operator part $\Theta_{\op}^0$ of $\Theta_\alpha^0$ is unitarily equivalent to the following pre-minimal difference operator $\rh_\alpha^0$ defined in $\ell^2(\cV)$ by
\be\label{eq:discr_lapl}
(\tau_{\cG,\alpha} f)(v) = \frac{1}{\sqrt{m(v)}}\left( \sum_{u\in \cV} b(v,u)\Big(\frac{f(v)}{\sqrt{m(v)}} - \frac{f(u)}{\sqrt{m(u)}}\Big)  + \frac{\alpha(v)}{\sqrt{m(v)}}f(v)\right),\quad v\in\cV,
\ee
where $b\colon \cV\times\cV \to [0,\infty)$ is  given by
\be\label{eq:b_weight}
b(v,u) = \begin{cases}  |e_{v,u}|^{-1}, & v\sim u,\\ 0, & v\not\sim u.\end{cases}
\ee
More precisely, we define the operator $\rh_\alpha^0$ in $\ell^2(\cV)$ on the domain $\dom(\rh_\alpha^0):=C_c(\cV)$ by 
\begin{align}\label{eq:h_alpha}
\begin{split}
\rh_\alpha^0\colon \ba[t]{lcl} \dom(\rh_\alpha^0) &\to&  \ell^2(\cV) \\ f &\mapsto& \tau_{\cG,\alpha}f \ea.
\end{split}
\end{align}
Here and below $C_c(\cV)$ is the space of finitely supported functions on $\cV$. 
 Notice that Hypothesis \ref{hyp:graph02} guarantees that $\rh_\alpha^0$ is well defined since  $\tau_{\cG,\alpha}f\in\ell^2(\cV)$ for every $f\in C_c(\cV)$. Moreover, $\rh_\gA^0$ is symmetric and let us denote its closure by $\rh_\gA$.

Thus we proved the following result.

\begin{proposition}\label{prop:bo_alpha}
Assume that Hypotheses \ref{hyp:graph01} and \ref{hyp:graph02} are satisfied. Also, let $\bH_\alpha$ be the closure of the pre-minimal operator \eqref{eq:Halpha} and let $\Pi_\cG$ be the boundary triplet \eqref{eq:tiPi}--\eqref{eq:tiG1}. Then 
\be\label{eq:}
\dom(\bH_\alpha) = \{f\in\dom(\bH_{\max})\colon\ \{\wt\Gamma_0f,\wt\Gamma_1f\}\in \Theta_\alpha\},
\ee
where $\Theta_\alpha$ is a linear relation in $\cH_{\cG}$ defined as the closure of $\Theta_\gA^0$ given by \eqref{eq:ThetaA}. Moreover, the operator part $\Theta_\alpha^{\op}$ of $\Theta_\alpha$ is unitarily equivalent to the operator $\rh_\alpha = \overline{\rh_\alpha^0}$ acting in $\ell^2(\cV)$. 
\end{proposition}

We also need another discrete Laplacian. Specifically, in the weighted Hilbert space $\ell^2(\cV;m)$ we consider the minimal operator defined by the following difference expression
\be\label{eq:h_a}
(\wt\tau_{\cG,\alpha} f)(v) := \frac{1}{m(v)}\left( \sum_{u\in \cV} b(v,u)(f(v) - f(u))  + \alpha(v)f(v)\right),\quad v\in\cV.
\ee 

\begin{lemma}\label{cor:similarity}
The pre-minimal operator $\tilde{\rh}_\alpha^0$ associated with \eqref{eq:h_a} in $\ell^2(\cV;m)$ is unitarily equivalent to the operator $\rh_\alpha^0$ defined by \eqref{eq:discr_lapl}, \eqref{eq:h_alpha} and acting in $\ell^2(\cV)$.
\end{lemma}
 
\begin{proof}
It suffices to note that 
\[
\tilde{\rh}_\alpha^0 =U^{-1}{\rh_\alpha^0}U, 
\]
where the operator 
\[
\begin{array}{cccc}
U\colon & \ell^2(\cV;m) & \to &\ell^2(\cV) \\
 & f & \mapsto & \sqrt{m}f
\end{array}
\]
isometrically maps $\ell^2(\cV;m)$ onto $\ell^2(\cV)$. 
\end{proof}
 
 In the following we shall use $\rh_\alpha$ as the symbol denoting the closures of both operators. 
Now we are ready to formulate our main result. 

\begin{theorem}\label{th:main}
Assume that Hypotheses \ref{hyp:graph01} and \ref{hyp:graph02} are satisfied. Let $\alpha\colon\cV\to \R$ and $\bH_\alpha$ be a closed symmetric operator associated with the graph $\cG$ and equipped with the $\delta$-type coupling conditions \eqref{eq:bcalpha} at the vertices. Also, let $\rh_\alpha$ be the discrete Laplacian defined either by \eqref{eq:discr_lapl} in $\ell^2(\cV)$ or by  \eqref{eq:h_a} in $\ell^2(\cV;m)$, where the functions $m\colon \cV\to (0,\infty)$ and $b\colon \cV\times\cV\to [0,\infty)$ are given by \eqref{eq:m_weight} and \eqref{eq:b_weight}, respectively. Then:
\begin{itemize}
\item[(i)] The deficiency indices of $\bH_\alpha$ and $\rh_\alpha$ are equal and 
\be\label{eq:def+=-}
\mathrm{n}_+(\bH_\alpha) = \mathrm{n}_-(\bH_\alpha) = \mathrm{n}_\pm(\rh_\alpha)\le \infty.
\ee
In particular, $\bH_\alpha$ is self-adjoint if and only if $\rh_\alpha$ is self-adjoint.
\end{itemize}
Assume in addition that $\bH_\alpha$ (and hence also $\rh_\alpha$) is self-adjoint. Then:
\begin{itemize}
\item[(ii)] The operator $\bH_\alpha$ is lower semibounded if and only if the operator $\rh_\alpha$ is lower semibounded. 
\item[(iii)] The operator $\bH_\alpha$ is nonnegative (positive definite) if and only if the operator $\rh_\alpha$ is nonnegative (respectively, positive definite). 
\item[(iv)] The total multiplicities of negative spectra of $\bH_\alpha$ and $\rh_\alpha$ coincide,
\be
\kappa_-(\bH_\alpha) = \kappa_-(\rh_\alpha).
\ee
\item[(v)] Moreover, the following equivalence 
\be
\bH_\alpha^- \in \gS_p(L^2(\cG)) \Longleftrightarrow \rh_\alpha^- \in \gS_p(\ell^2(\cV;m)),
\ee
holds for all $p\in (0,\infty]$. In particular, negative spectra of $\bH_\alpha$ and $\rh_\alpha$ are discrete simultaneously.
\item[(vi)] If $\rh_\alpha^- \in \gS_\infty(\ell^2(\cV;m))$, then the following equivalence holds for all $\gamma\in(0,\infty)$ 
\be\label{eq:weakSp_a}
\lambda_j(\bH_{\alpha}) =  j^{-\gamma}(a+o(1)) \ \Longleftrightarrow \ \lambda_j(\rh_\alpha)=  j^{-\gamma}(b+o(1)),
\ee
as $j\to \infty$, where either $ab\not = 0$ or $a = b = 0$.
\item[(vii)] If, in addition, $\rh_\alpha$ is lower semibounded, then $\inf\sigma_{\ess}(\bH_\alpha)>0$ $(\inf\sigma_{\ess}(\bH_\alpha)=0)$ exactly when $\inf\sigma_{\ess}(\rh_\alpha)>0$ $($respectively, $\inf\sigma_{\ess}(\rh_\alpha)=0)$.
\item[(viii)] The spectrum of $\bH_\alpha$ is purely discrete if and only if the number $\#\{e\in\cE\colon |e|>\varepsilon\}$ is finite for every $\varepsilon>0$ and the spectrum of the operator $\rh_\alpha$ is purely discrete.
\item[(ix)] If $\ti{\gA}\colon \cV\to \R$ is such that $\rh_{\ti\alpha}=\rh_{\ti\alpha}^\ast$, then the following equivalence 
\be
(\bH_\alpha - \I)^{-1} - (\bH_{\ti\alpha} - \I)^{-1} \in \gS_p(L^2(\cG)) \Longleftrightarrow (\rh_\alpha - \I)^{-1} - (\rh_{\ti\alpha} - \I)^{-1} \in \gS_p(\ell^2(\cV)),
\ee
holds for all $p\in (0,\infty]$. 
\end{itemize}
\end{theorem}

\begin{proof}
We only need to comment on the first equality in \eqref{eq:def+=-} since the rest immediately follows from Theorem \ref{th:HTheta} and Proposition \ref{prop:bo_alpha}. However, the first equality in \eqref{eq:def+=-} follows from the equality of deficiency indices of the operator $\rh_\gA$. Indeed, $\mathrm{n}_+(\rh_\gA)=\mathrm{n}_-(\rh_\gA)$ by the von Neumann theorem since $\rh_\gA$ commutes with the complex conjugation. 
\end{proof}

Let us demonstrate Theorem \ref{th:main} by applying it to the 1-D Schr\"odinger operator with $\delta$-interactions \eqref{eq:1Ddelta} considered in Example \ref{ex:01}.

\begin{example}\label{ex:02}
Let  $\rH_{X,\gA}$ be the Schr\"odinger operator \eqref{eq:1Ddelta} with
$\delta$-interactions on the positive semi-axis $(0,\infty)$. Recall that in this case $\cV=\{x_k\}_{k\ge 0}$ and $\cE = \{e_k\}_{k\in\N}$, where $e_k=(x_{k-1},x_k)$. By \eqref{eq:m_weight} and \eqref{eq:b_weight}, we get
\[
m(x_k)=\begin{cases} |e_1|, & k=0,\\
|e_{k}|+|e_{k+1}|, & k\in\N,
\end{cases}
\]
where $|e_k| = x_k - x_{k-1}$ for all $k\in\N$, and 
\[
b(x_k,x_n) = \begin{cases} |x_k-x_n|^{-1}, &
|n-k|=1,\\0, & |n-k|\neq 1.\end{cases}
\]
Setting $f=\{f_k\}_{k\ge0}$ with $f_k:=f(x_k)$, we see that the difference expression \eqref{eq:discr_lapl} is just a three-term recurrence relation
\[
(\tilde{\tau}_{\alpha} f)_k = \begin{cases} b_1(f_0 - f_1), & k=0,\\ 
-b_k f_{k-1} + a_k f_k - b_{k+1}f_{k+1}, & k\in\N, \end{cases}
\]
where
\begin{align}\label{eq:akbk}
a_k & = \frac{\alpha_k + |e_{k}|^{-1} + |e_{k+1}|^{-1}}{m(x_k)}, & b_k & =  \frac{|e_{k}|^{-1}}{\sqrt{m(x_{k-1})m(x_k)}},
\end{align}
for all $k\in\N$. 
Hence the corresponding operator $\rh_\alpha$ is the minimal operator associated in $\ell^2(\Z_{\ge 0})$ with the Jacobi (tri-diagonal) matrix
\be\label{eq:1Djacop}
J= \begin{pmatrix} b_1 & -b_1 & 0 & 0 & \dots \\
-b_1 & a_1 & -b_2 &  0 & \dots \\
0& -b_2 & a_2 & -b_3 & \dots \\
 0 & 0 & -b_3 & a_3 & \dots \\
\dots & \dots & \dots & \dots & \dots \\
\end{pmatrix}.
\ee
In this particular case Theorem \ref{th:main} was established in \cite{KM10} and in the recent paper \cite{kmn16} it was  extended to the case of Schr\"odinger operators in a space of vector-valued functions. \hfill $\lozenge$
\end{example}

\begin{remark}\label{rem:difJacobi}
Let us emphasize the difference between the operators generated by \eqref{eq:1.03} and \eqref{eq:1.03B} in the case $\inf_{e\in\cE}|e|=0$. Indeed, replacing $m$ by $\deg$ in \eqref{eq:akbk} and noting that $\deg(x_k)=2$ for all $k\in\N$, we end up with the Jacobi matrix, which does not reflect spectral properties of the Hamiltonian $\rH_{X,\gA}$. For example, setting $|e_k|=1/k$, $k\in\N$,  \eqref{eq:discr_lapl} with $\deg$ in place of $m$ then gives rise to the matrix
\be\label{eq:1Djacoptilde}
\wt{J}= \frac{1}{2}\begin{pmatrix} \sqrt{2} & -\sqrt{2} & 0 & 0 & \dots \\
-\sqrt{2} & \gA_1 + 3 & -{2} &  0 & \dots \\
0& -{2} & \gA_2 + 5 & -3 & \dots \\
 0 & 0 & -3 & \gA_3 + 7 & \dots \\
\dots & \dots & \dots & \dots & \dots \\
\end{pmatrix}.
\ee
The Carleman test shows that the minimal operator associated with $\wt{J}$ in $\ell^2(\Z_{\ge 0})$ is always self-adjoint, however, $J$ with $\alpha_k:=-(2k+1)$ for all $k\in \N$ defines in $\ell^2(\Z_{\ge 0})$ the minimal symmetric operator with deficiency indices $(1,1)$ (cf. Example \ref{ex:01}). In particular, in this case the spectrum of every self-adjoint extension of $J$ (and hence of $\rH_{X,\gA}$!) is purely discrete, however, the spectrum of $\wt{J}$ with this choice of $\alpha$ is purely absolutely continuous and covers the whole real line $\R$ (cf. \cite{jn02}). The latter shows that one cannot replace \eqref{eq:1.03} by \eqref{eq:1.03B} in Theorem \ref{th:main} if $\inf_{e\in\cE}|e|=0$.
\end{remark}

\begin{remark}\label{rem:3.8}
One can notice a connection between the discrete Laplacian \eqref{eq:h_a} and the operator $\bH_\alpha$ without the boundary triplets approach. Namely, consider the kernel $\cL= \ker(\bH_{\max})$ of $\bH_{\max}$, which consists of piecewise linear functions on $\cG$. Every $f\in \cL$ can be identified with its values  $\{f(e_i), f(e_o)\}_{e\in \cE}$ on $\cV$. First of all, notice that 
\be
\|f\|^2_{L^2(\cG)} = \sum_{e\in\cE} |e| \frac{|f(e_i)|^2 + \re(f(e_i)f(e_o)^\ast) + |f(e_o)|^2}{3}.
\ee
Now restrict ourselves to the subspace $\cL_{cont}$ of $\cL$ which consists of continuous functions vanishing everywhere on $\cG$ except finitely many edges. Clearly,  
\[
 \sum_{e\in\cE} |e| ({|f(e_i)|^2 + |f(e_o)|^2}) = \sum_{v\in\cV} |f(v)|^2\sum_{e\in \cE_v}|e|= \|f\|^2_{\ell^2(\cV;m)}
\]
defines an equivalent norm on $\cL_{cont}$. 
On the other hand, for every $f\in \cL_{cont}$ we get
\begin{align*}
(\bH_\alpha f,f) & = \sum_{e\in\cE} \int_{e} |f'({x_e})|^2 d{x_e} + \sum_{v\in\cV} \alpha(v)|f(v)|^2\\
& = \sum_{e\in\cE} \frac{|f(e_o) - f(e_i)|^2}{|e| } + \sum_{v\in\cV} \alpha(v)|f(v)|^2\\
&=\frac{1}{2}\sum_{u,v\in \cV} b(v,u)|f(v) - f(u)|^2 + \sum_{v\in\cV}\alpha(v)|f(v)|^2=:\gt_{\cG,\alpha}[f].
\end{align*}
However, one can easily check that the latter is the quadratic form of the discrete operator $\rh_\alpha$ defined in $\ell^2(\cV;m)$ by \eqref{eq:h_a}, that is, the following equality 
\be\label{eq:hAform}
\big(\rh_\alpha f,f\big)_{\ell^2(\cV;m)} = \gt_{\cG,\alpha}[f]=\frac{1}{2}\sum_{u,v\in \cV} b(v,u)|f(v) - f(u)|^2 + \sum_{v\in\cV}\alpha(v)|f(v)|^2
\ee
holds for every $f\in C_c(\cV)$. 
\end{remark}

\section{Quantum graphs with Kirchhoff vertex conditions}\label{sec:kirchhoff}

As in Section \ref{sec:parameterization}, if it is not explicitly stated, we shall always assume that $\cG$ satisfies Hypotheses \ref{hyp:graph01} and \ref{hyp:graph02}.
In this section we restrict ourselves to the case $\alpha\equiv 0$, that is, we consider the quantum graph with Kirchhoff vertex conditions 
\be\label{eq:kirchhoff}
\begin{cases} f\ \text{is continuous at}\ v,\\[2mm] \sum_{e\in \cE_v}f_e'(v) =0, \end{cases}
\ee 
at every vertex $v\in\cV$. 
Let us denote  by $\bH_0$ the closure of the corresponding operator $\bH_0^0$ given by \eqref{eq:Halpha}. By Theorem \ref{th:main}, the spectral properties of $\bH_0$ are closely connected with those of $\rh_0$, where $\rh_0$ is the discrete Laplacian defined in $\ell^2(\cV;m)$ by the difference expression
\be\label{eq:h0}
(\tau_{\cG,0} f)(v) = \frac{1}{m(v)} \sum_{u\sim v} b(u,v)(f(v) - f(u)), \quad v\in\cV, 
\ee 
and the functions $m\colon\cV\to (0,\infty)$, $b\colon\cV\times\cV \to [0,\infty)$ are defined by \eqref{eq:m_weight} and \eqref{eq:b_weight}, respectively,
\begin{align}\label{eq:mb}
m(v) & =\sum_{e\in\cE_v}|e|, & b(u,v) & = \begin{cases}  |e_{u,v}|^{-1}, & u\sim v,\\ 0, & u\not\sim v.\end{cases}
\end{align}
Note that both operators $\bH_0$ and $\rh_0$ are symmetric and nonnegative. Moreover, Theorem \ref{th:main} immediately implies the following result.

\begin{corollary}\label{cor:kirchhoff}
Assume that Hypotheses \ref{hyp:graph01} and \ref{hyp:graph02} are satisfied. 
Then: 
\begin{itemize}
\item[(i)]
The deficiency indices of $\bH_0$ and $\rh_0$ are equal and
\[
\mathrm{n}_+(\bH_0) = \mathrm{n}_-(\bH_0) = \mathrm{n}_\pm(\rh_0)\le \infty.
\]
In particular, $\bH_0$ is self-adjoint if and only if $\rh_0$ is self-adjoint.
\end{itemize}

Assume in addition that $\bH_0$ (and hence also $\rh_0$) is self-adjoint. Then:
\begin{itemize}
\item[(ii)]  $\bH_0$ is positive definite if and only if the same is true for $\rh_0$.
\item[(iii)]  $\inf\sigma_{\ess}(\bH_0)>0$  if and only if $\inf \sigma_{\ess}(\rh_0)>0$.
\item[(iv)] The spectrum of $\bH_0$ is purely discrete if and only if the number $\#\{e\in\cE\colon |e|>\varepsilon\}$ is finite for every $\varepsilon>0$ and the spectrum of the operator $\rh_0$ is purely discrete.
\end{itemize}
\end{corollary}

Our next goal is to use the spectral theory of discrete Laplacians \eqref{eq:h0} to prove new results for quantum graphs. 

\subsection{Intrinsic metrics on graphs} \label{ss:4.1}

During the last decades a lot of attention has been paid to the study of spectral properties of the discrete Laplacian \eqref{eq:h0}. Let us recall several basic concepts. Suppose that the metric graph $\cG=(\cV,\cE,|\cdot|)$ satisfies Hypotheses \ref{hyp:graph01} and \ref{hyp:graph02}. The function $\Deg\colon \cV \to (0,\infty)$ defined by 
 \be\label{eq:wgtdeg}
\Deg\colon v\mapsto \frac{1}{m(v)}\sum_{u\in\cE_v} b(u,v) = \frac{\sum_{e\in\cE_v} |e|^{-1}}{\sum_{e\in\cE_v} |e|},
\ee
is called {\em the weighted degree}.  Notice that by \cite[Lemma 1]{dav} (see also \cite[Theorem 11]{kl10}), $\rh_0$ is bounded on $\ell^2(\cV;m)$ (and hence self-adjoint) if and only if the weighted degree $\Deg$ is bounded on $\cV$. In this case (see \cite[Lemma 1]{dav})
\be\label{eq:estDavies}
\sup_{v\in\cV} \Deg(v) \le \|\rh_0\|\le 2\sup_{v\in\cV} \Deg(v).
\ee

{\em A pseudo metric} $\varrho$ on $\cV$ is a symmetric function $\varrho\colon \cV\times\cV\to [0,\infty)$ such that $\varrho(v,v)=0$ for all $v\in\cV$ and satisfies the triangle inequality. Notice that every function $p\colon \cE \to (0,\infty)$ generates {\em a path pseudo metric} $\varrho_p$ on $\cV$ with respect to the graph $\cG$ via
\be\label{eq:pathmetric}
\varrho_p(u,v) := 
\inf_{\cP=\{v_0,\dots,v_n\}\colon u=v_0,\ v=v_n}\sum_{k} p(e_{v_{k-1},v_k}). 
\ee
Here the  infimum is taken over all paths connecting $u$ and $v$.

Following \cite{flw14} (see also \cite{bkw15, kel15}), a pseudo metric $\varrho$ on $\cV$ 
 is called {\em intrinsic} with respect to the graph $\cG$ if 
\be\label{eq:intrinsicdef}
 \sum_{u\in\cE_v}b(u,v)\varrho(u,v)^2\le m(v)
\ee
holds on $\cV$. 
Notice that for any given locally finite graph an intrinsic metric always exists. 

\begin{example}
\begin{itemize}
\item[(a)]
Let  $p\colon \cE\to (0,\infty)$ be defined by
\be\label{eq:vrho0}
 p\colon e_{u,v}\mapsto \big(\Deg(u)\vee \Deg(v)\big)^{-1/2}. 
\ee
It is straightforward to check that the corresponding path pseudo metric $\varrho_p$ is intrinsic (see \cite[Example 2.1]{hkmw13}, \cite{kel15}). 

\item[(b)] Another pseudo metric was suggested in \cite{dVTHT}. Namely, let $\varrho$ be a path pseudo metric generated by the function $p\colon \cE\to (0,\infty)$
\be
p\colon e_{u,v}\mapsto  \left(\frac{m(u)\wedge m(v)}{b(e_{u,v})}\right)^{1/2}.
\ee
It was shown in \cite{hkmw13} that this metric is equivalent to the metric \eqref{eq:vrho0} if and only if the combinatorial degree $\deg$ is bounded on $\cV$. \hfill $\lozenge$
\end{itemize}
\end{example}

It turns out that for the discrete operator $\rh_0$ given by \eqref{eq:h0}, \eqref{eq:mb} the natural path metric induced by the metric graph $\cG$ is intrinsic.

\begin{lemma}\label{lem:varrho0}
The function $p_0\colon \cE\to (0,\infty)$ given by 
\be\label{eq:p0}
p_0(e):= |e|,\quad e\in\cE,
\ee
generates an intrinsic (with respect to the graph $\cG$) path metric $\varrho_0$ on $\cV$. 
\end{lemma}

\begin{proof}
First of all, notice that for the functions \eqref{eq:mb} the condition \eqref{eq:intrinsicdef} takes the following form
\be\label{eq:intrinsic}
\sum_{u\sim v}\frac{ \varrho(u,v)^2}{ |e_{u,v}|} \le \sum_{u\sim v}|e_{u,v}| 
\ee
for every $v\in\cV$. Clearly  \eqref{eq:intrinsic} holds with $\varrho=\varrho_0$ for all $v\in\cV$ with equality instead of inequality since 
\[
\varrho_0(u,v) = \frac{1}{b(u,v)} = |e_{u,v}|
\]
whenever $u\sim v$. 
\end{proof}

For any $v\in\cV$ and $r\ge 0$, {\em the distance ball} $B_r(v)$ with respect to a pseudo metric $\varrho$ is defined by 
\be\label{eq:ball}
B_r(v) : =\{u\in\cV\colon \varrho(u,v)\le r\}.
\ee
Finally for a set $X\subset \cV$, {\em the combinatorial neighbourhood} of $X$ is given by
\be\label{eq:comb_ball}
\Omega(X) :=\{u\in \cV\colon u\in X\ \text{or there exists}\ v\in X\ \text{such that}\ u\sim v\}.
\ee

\subsection{Self-adjointness of $\bH_0$}\label{ss:4.2}

In this and the following subsections we shall always assume that the metric graph $\cG$ satisfies Hypotheses \ref{hyp:graph01} and \ref{hyp:graph02}. We begin with the following result.

\begin{theorem}\label{th:saH0_1} 
If the weighted degree $\Deg$ is bounded on $\cV$,
\be\label{eq:4.01}
\sup_{v\in\cV}\Deg(v) = \sup_{v\in\cV} \frac{\sum_{e\in\cE_v} |e|^{-1}}{\sum_{e\in\cE_v} |e|}<\infty,
\ee
then the operator $\bH_0$ is self-adjoint. 
\end{theorem}

\begin{proof}
Consider the corresponding boundary operator $\rh_0$ defined by \eqref{eq:h0}. Since $\Deg$ is bounded on $\cV$, the operator $\rh_0$ is bounded on $\ell^2(\cV;m)$  (see \eqref{eq:estDavies}) and hence self-adjoint. It remains to apply Corollary \ref{cor:kirchhoff}(i).
\end{proof}

As an immediate corollary of this result we obtain the following widely known sufficient condition (cf. \cite[Theorem 1.4.19]{bk13}).
 
\begin{corollary}\label{cor:sa01}
If $\inf_{e\in\cE} |e|>0$, then the operator $\bH_0$ is self-adjoint.
\end{corollary}

\begin{proof}
By Theorem \ref{th:saH0_1}, it suffices to check that $\Deg$ is bounded on $\cV$:
\[
\sup_{v\in\cV} \frac{\sum_{e\in\cE_v} |e|^{-1}}{\sum_{e\in\cE_v} |e|} \le \sup_{v\in\cV} \frac{\deg(v) (\inf_{e\in\cE} |e|)^{-1}}{\deg(v) \inf_{e\in\cE} |e|} = \frac{1}{(\inf_{e\in\cE} |e|)^2}<\infty. \qedhere
\]
\end{proof} 

 A few remarks are in order:

\begin{remark}
\begin{itemize}
\item[(i)] 
Numerous graphs considered both in theoretical purposes and in applications belong to this category \cite{bk13}. Prominent examples are \emph{equilateral graphs} (see, e.g., \cite{cat, pan12, pan13}) and \emph{periodic graphs} (with a finite number of edges in the period cell).
\item[(ii)] Notice that under Hypothesis \ref{hyp:graph02}, the conditions $ \inf_{e\in\cE} |e|>0$ and \eqref{eq:4.01}  are equivalent only if $\sup_{v\in\cV}\deg(v)<\infty$. It is not difficult to construct examples of graphs such that $ \inf_{e\in\cE} |e|=0$ and condition \eqref{eq:4.01} is satisfied (see Example \ref{ex:tree01} below). 
\end{itemize}
\end{remark}

\begin{example}\label{ex:tree01}
Let $\{n_k\}_{k\in \N}$ be a strictly increasing sequence of natural numbers. 
 Consider the following metric graph: Let $o$ be a distinguished vertex which has $n_1$ emanating edges. Moreover, suppose that one of those edges has length $\frac{1}{n_1}$ and the other edges have a fixed length, say $1$. Next, suppose every vertex in the first combinatorial sphere (i.e., every $v\sim o$) has $n_2$ emanating edges and again their lengths equal $1$ except one edge having length $\frac{1}{n_2}$. Continuing this procedure to infinity we end up with an infinite metric graph (called a {\em rooted tree}) such that  
 \begin{align*}
 \inf_{e\in\cE} |e| & = \inf_{k\ge 1} \frac{1}{n_k}=0, & \sup_{e\in\cE} |e| & =1. 
 \end{align*}
 It is easy to see that
\begin{align*}
\sup_{v\in\cV}\Deg(v) &= \sup_{v\in\cV} \frac{\sum_{e\in\cE_v} |e|^{-1}}{\sum_{e\in\cE_v} |e|}  = \sup_{k\ge 1} \frac{n_{k+1}-1 + n_k+n_{k+1}}{n_{k+1} - 1 + \frac{1}{n_k}+\frac{1}{n_{k+1}}} < 4.
\end{align*}
Hence, by Theorem \ref{th:saH0_1} the corresponding Hamiltonian $\bH_0$ is self-adjoint. Moreover, we shall prove below (see Lemma \ref{lem:sa01}) that in this case the corresponding Hamiltonian $\bH_\alpha$ with $\delta$ interactions is self-adjoint for any $\alpha\colon \cV\to \R$. 
\hfill $\lozenge$
\end{example}

The next result shows that we can replace uniform boundedness of the weighted degree function by the local one (in a suitable sense of course).

\begin{theorem}\label{th:gaffney}
Let $\varrho$ be an intrinsic pseudo metric on $\cV$ such that the weighted degree $\Deg$ is bounded on every distance ball in $\cV$. Then $\bH_0$ is self-adjoint.
\end{theorem}

\begin{proof}
By \cite[Theorem 1]{hkmw13}, the operator $\rh_0$ is self-adjoint. Hence by Corollary \ref{cor:kirchhoff}(i) so is $\bH_0$.
\end{proof}

As an immediate corollary we arrive at the following Gaffney type theorem for quantum graphs.  

\begin{corollary}\label{cor:gaffney}
Let $\varrho_0$ be a natural path metric on $\cV$ defined in Lemma \ref{lem:varrho0}. If $(\cV,\varrho_0)$ is complete as a metric space, then $\bH_0$ is self-adjoint.
\end{corollary}

\begin{proof}
By Hypothesis \ref{hyp:graph02}, the discrete  graph $\cG_d=(\cV,\cE)$ is locally finite. Hence by a Hopf--Rinow type theorem \cite{hkmw13}, $(\cV,\varrho_0)$ is complete as a metric space if and only if the distance balls in  $(\cV,\varrho_0)$ are finite. The latter immediately implies that the weighted degree  $\Deg$ is bounded on every distance ball in $(\cV,\varrho_0)$. It remains to apply Theorem \ref{th:gaffney}.
\end{proof}

\begin{remark}\label{rem:4.10}
Notice that Corollary \ref{cor:gaffney} can be seen as the analog of the classical result of Gaffney \cite{gaf} (see also \cite[Chapter 11]{gri} for further details), who established self-adjointness of the Dirichlet Laplacian on a complete Riemannian manifold. Indeed, $|\cdot|$ generates a natural path metric on a metric graph $\cG=(\cV,\cE,|\cdot|)$ and it is easy to check that $\cG$ equipped with this metric is complete as a metric space if and only if $(\cV,\varrho_0)$ is complete as a metric space.

Let us also mention that Corollary \ref{cor:gaffney} proves the self-adjointness of $\bH_0$ if the metric graph $\cG$ satisfies the {\em finite ball condition} (see \cite[Assumption 1.3.5]{bk13}), which is equivalent to the completeness of $(\cV,\varrho_0)$.
\end{remark}

On the one hand, simple examples demonstrate that Corollary \ref{cor:gaffney} is sharp. Indeed, consider the second derivative on an interval $(0,\ell)$ with $\ell\in (0, \infty]$. As in Example \ref{ex:01}, let $\{x_k\}_{k\ge 0}$ be a strictly increasing sequence such that $x_k\uparrow \ell$ as $k\to \infty$. In this case Kirchhoff conditions are equivalent to the continuity of a function and its derivative at every vertex $x_k$ (see \eqref{eq:delta_line}). The corresponding  operator is self-adjoint only if $\ell=\infty$. 
However, we can improve Corollary \ref{cor:gaffney} by replacing the natural path metric $\varrho_0$ by another path metric (which is not intrinsic!) generated by the weight function $m$. 

\begin{theorem}\label{prop:gaffney}
Let $p_m\colon \cE\to (0,\infty)$ be defined by 
\be\label{eq:p_m}
p_m\colon e_{u,v} \mapsto m(u)+m(v), 
\ee 
where $m$ is given by \eqref{eq:mb}, and let $\varrho_m$ be the corresponding path metric \eqref{eq:pathmetric}.  If $(\cV,\varrho_m)$ is complete as a metric space, then 
$\bH_0$ is self-adjoint.
\end{theorem}

\begin{proof}
Applying the Hopf--Rinow theorem from \cite{hkmw13} once again, $(\cV,\varrho_m)$ is complete as a metric space if and only if all infinite geodesics have infinite length, which is further equivalent to the fact that distance balls in  $(\cV,\varrho_m)$ are finite. The former statement implies, in particular, that for every infinite path $\cP = \{v_n\}_{n\ge 0}\subset \cV$ its length 
\[
|\cP| = \sum_{n\ge 0} p_m(e_{v_n,v_{n+1}}) 
\]
is infinite. However, \eqref{eq:p_m} implies the following estimate
\[
\sum_{n=0}^N m(v_n) \le \sum_{n=0}^N p_m(e_{v_n,v_{n+1}}) \le 2 \sum_{n=0}^N m(v_n),
\]
for every finite path $\cP_N=\{v_n\}_{n=0}^N$ in $\cV$. Hence for every infinite path $\cP$ we conclude that the sum 
\[
\sum_{n\ge 0} m(v_n) 
\]
is infinite. By Theorem 6 from \cite{kl12}, the latter implies that the operator $\rh_0$ is self-adjoint in $\ell^2(\cV;m)$. It remains to apply Corollary \ref{cor:kirchhoff}(i).
\end{proof}

As an immediate corollary of Theorem \ref{prop:gaffney} we obtain the following improvement of Corollary \ref{cor:sa01}.

\begin{corollary}\label{lem:sa02}
If 
\be\label{eq:4.12}
\inf_{v\in\cV} m(v) = \inf_{v\in\cV} \sum_{e\in\cE_v}|e|>0,
\ee
 then the operator $\bH_0$ is self-adjont.
\end{corollary}

\begin{proof}
Clearly, every infinite geodesic in $(\cV,\varrho_m)$ has infinite length if \eqref{eq:4.12} is satisfied. According to Hypothesis \ref{hyp:graph02}, $\cG$ is a locally finite graph and hence combining the Hopf--Rinow type theorem \cite{hkmw13} with Theorem \ref{prop:gaffney} we finish the proof.
\end{proof}

\begin{remark}
\begin{itemize}
\item[(i)] Notice that the self-adjointness of $\rh_0$ in $\ell^2(\cV;m)$ under the assumption \eqref{eq:4.12} was first  mentioned in \cite[Corollary 9.2]{hklw12}.
\item[(ii)] Clearly, $\varrho_{0}(u,v)\le \varrho_{m}(u,v)$ for all $u,v\in\cV$ and hence every infinite geodesic in $(\cV,\varrho_0)$ with infinite length will have an infinite length in $(\cV,\varrho_m)$. However, the converse statement is not true which can be seen by simple examples.
\end{itemize}
\end{remark}

\begin{example}\label{ex:4.14}
Let $\cG\subset \R^2$ be a planar graph constructed as follows (see the figure depicted below). Let $X=\{x_k\}_{k\ge 1} \subset [0,\infty)$ be a strictly increasing sequence with $x_1=0$. We set $\cV = X\times \{-1,0,1\}$ and denote $v_{k,n} = (x_k,n)$, $k\in\N$ and $n\in\{-1,0,1\}$. Now we define the set of edges by the following rule: $v_{n,k} \sim v_{m,j}$ if either $n=m$ and $|k-j|=1$ or $k=j=0$ and $|n-m|=1$. Finally, we assign lengths as the usual Euclidean length in $\R^2$: the length of every vertical edge is equal to $1$, and the length of the horizontal edge $e_{v_{k,0}, v_{k+1,0}}$ is equal to  $x_{k+1}-x_k$.

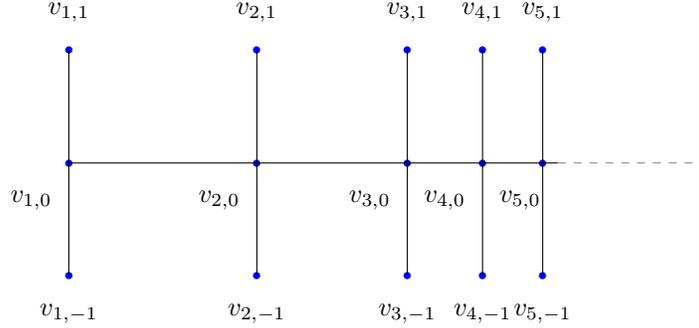
\begin{center}
 \begin{tikzpicture}[domain=-4:8, samples=101]
 
\node[color=blue] at (-3,0) {\tiny $\bullet$};
\node at (-3.5,-0.5) {$v_{1,0}$};

\node[color=blue] at (-0.5,0) {\tiny $\bullet$};
\node at (-1,-0.5) {$v_{2,0}$};

\node[color=blue] at (1.5,0) {\tiny $\bullet$};
\node at (1,-0.5) {$v_{3,0}$};

\node[color=blue] at (2.5,0) {\tiny $\bullet$};
\node at (2,-0.5) {$v_{4,0}$};

\node[color=blue] at (3.3,0) {\tiny $\bullet$};
\node at (3,-0.5) {$v_{5,0}$};

\node[color=blue] at (-3,1.5) {\tiny $\bullet$};
\node at (-3,2) {$v_{1,1}$};

\node[color=blue] at (-0.5,1.5) {\tiny $\bullet$};
\node at (-0.5,2) {$v_{2,1}$};

\node[color=blue] at (1.5,1.5) {\tiny $\bullet$};
\node at (1.5,2) {$v_{3,1}$};

\node[color=blue] at (2.5,1.5) {\tiny $\bullet$};
\node at (2.5,2) {$v_{4,1}$};

\node[color=blue] at (3.3,1.5) {\tiny $\bullet$};
\node at (3.3,2) {$v_{5,1}$};

\node[color=blue] at (-3,-1.5) {\tiny $\bullet$};
\node at (-3,-2) {$v_{1,-1}$};

\node[color=blue] at (-0.5,-1.5) {\tiny $\bullet$};
\node at (-0.5,-2) {$v_{2,-1}$};

\node[color=blue] at (1.5,-1.5) {\tiny $\bullet$};
\node at (1.5,-2) {$v_{3,-1}$};

\node[color=blue] at (2.5,-1.5) {\tiny $\bullet$};
\node at (2.5,-2) {$v_{4,-1}$};

\node[color=blue] at (3.3,-1.5) {\tiny $\bullet$};
\node at (3.3,-2) {$v_{5,-1}$};

\draw[-] (-3,0) -- (3.5,0);
\draw[dashed, help lines] (3.5,0) -- (5.5,0);

\draw[-] (-3,0) -- (-3,1.5);
\draw[-] (-0.5,0) -- (-0.5,1.5);
\draw[-] (1.5,0) -- (1.5,1.5);
\draw[-] (2.5,0) -- (2.5,1.5);
\draw[-] (3.3,0) -- (3.3,1.5);

\draw[-] (-3,0) -- (-3,-1.5);
\draw[-] (-0.5,0) -- (-0.5,-1.5);
\draw[-] (1.5,0) -- (1.5,-1.5);
\draw[-] (2.5,0) -- (2.5,-1.5);
\draw[-] (3.3,0) -- (3.3,-1.5);

\end{tikzpicture}
\end{center}

Clearly, $(\cV,\varrho_0)$ is complete as a metric space if and only if 
\[
\sum_{k\ge 0}|e_{v_{k,0}, v_{k+1,0}}|=\sum_{k\ge 1} (x_{k+1}-x_k) =\lim_{k\to \infty} x_k =\infty.
\]
On the other hand, 
\[
m(v) = \sum_{u\sim v}|e_{u,v}|\ge 1
\]
for all $v\in\cV = X\times \{-1,0,1\}$, and hence $(\cV,\varrho_m)$ is always complete.  Therefore, the corresponding operator $\bH_0$ is always self-adjoint in view of Corollary \ref{lem:sa02}.
\hfill$\lozenge$
\end{example}

\begin{remark}
The graphs considered in Examples \ref{ex:tree01} and \ref{ex:4.14} belong to a special class of graphs, the so-called trees. More precisely, a path $\cP=\{v_0,v_1,\dots, v_n\}\subset \cV$  is called {\em a cycle} if $v_0=v_n$. A connected graph $\cG_d=(\cV,\cE)$ without cycles is called {\em a tree}. Notice that for any two vertices $u$, $v$ on a tree $\cT = (\cV, \cE)$ there is exactly one path $\cP$ connecting $u$ and $v$ and hence every path on a tree is a geodesic with respect to a path metric. 
\end{remark}

Let us finish this subsection with some sufficient conditions for $\bH_0$ to have nontrivial deficiency indices. Let $\varrho_{1/2}$ be a path metric on $\cV$ generated by the function $p_{1/2}\colon \cE\to (0,\infty)$ defined by 
\be
p_{1/2}\colon e\mapsto\sqrt{|e|}. 
\ee
If $(\cV,\varrho_{1/2})$ is not complete as a metric space, we then denote the metric completion of $(\cV,\varrho_{1/2})$ by $\overline{\cV}$  and $\cV_\infty:=\overline{\cV}\setminus\cV$. By \cite[Lemma 2.1]{dVTHT}, every function $f\colon \cV\to\R$ such that the corresponding quadratic form
\[
\gt_{\cG,0}[f] = \frac{1}{2}\sum_{u,v\in \cV} b(v,u)|f(v) - f(u)|^2 
\]
is finite, is uniformly Lipschitz with respect to the metric $\varrho_{1/2}$ and hence admits a continuation $F$ to $\overline{\cV}$ as a Lipschitz function. Following \cite{dVTHT}, we set $f_\infty:=F\upharpoonright \cV_\infty$. 

\begin{proposition}\label{th:rho1}
If $(\cV,\varrho_{1/2})$ is not complete as a metric space and there is $f\colon \cV\to \R$ such that $\gt_{\cG,0}[f]<\infty$ and $f_\infty\neq 0$, then $\bH_0$ is not a self-adjoint operator.
\end{proposition}

\begin{proof}
Follows from \cite[Theorem 3.1]{dVTHT} and Corollary \ref{cor:kirchhoff}(i).
\end{proof}

A few remarks are in order.

\begin{remark}
\begin{itemize}
\item[(i)]
The question on deficiency indices of $\rh_0$ in this case was left in \cite{dVTHT} as an open problem. 
\item[(ii)] 
Clearly, Proposition \ref{th:rho1} provides only a sufficient condition for $\bH_0$ to have nontrivial deficiency indices.
\end{itemize}
\end{remark}

\begin{example}\label{ex:4.14cntd}
Let us slightly modify the metric graph considered in Example \ref{ex:4.14} by shrinking the vertical edges. 
It is not difficult to show (see, e.g., \cite{car00,car11}) that the corresponding operator $\bH_0$ is not self-adjoint if the graph $\cG$ has finite total length,
\be\label{eq:finitevolume}
\sum_{e\in\cE}|e| <\infty.
\ee
On the other hand, the latter is further equivalent to the fact that $(\cV;\varrho_m)$ is not complete as a metric space. Thus, Theorem \ref{prop:gaffney} provides a self-adjointness criterion in this case. 

Let us also mention that we expect that the deficiency indices of the operator $\bH_0$ in the case \eqref{eq:finitevolume} are equal to one. \hfill$\lozenge$
\end{example}

\begin{remark}
Taking into account the above example, it is a rather natural guess that Theorem \ref{prop:gaffney} provides a self-adjointness criterion not only in the special case considered in Example \ref{ex:4.14cntd} but also for arbitrary graphs. However, for radially symmetric trees, the operator $\bH_0$ is not self-adjoint exactly when the corresponding tree has finite total length, that is, \eqref{eq:finitevolume} holds true (see \cite[\S 3.4]{sol02} and also \cite{car00}). Moreover, it is easy to check that in this case \eqref{eq:finitevolume} is not equivalent to non-completeness of $(\cV;\varrho_m)$.
\end{remark}
\subsection{Uniform positivity and the essential spectrum of $\bH_0$}\label{ss:4.3} 

For any vertex set  $X\subset \cV$, {\em the boundary } $\partial X$ of $X$ is defined by 
\be
\partial X:= \{(u,v)\in X\times (\cV\setminus X)\colon u\sim v\}.
\ee
For every subset $\wt{\cV}\subseteq \cV$ one defines {\em the isoperimetric constant} 
\be\label{eq:isoper}
C(\wt{\cV}) := \inf_{X\subset \wt{\cV}} \frac{\#(\partial X)}{m(X)},
\ee
where 
\begin{align}
\#(\partial X) & = \sum_{(u,v)\in\partial X} 1, & m(X) & = \sum_{v\in X} m(v) = \sum_{v\in X}\sum_{e\in\cE_v} |e|.
\end{align}
Moreover, we need {\em the isoperimetric constant at infinity}
\be\label{eq:isoperess}
C_{\ess}(\cV) := \sup_{X\subset \cV\ \text{is finite}} C(\cV\setminus X).
\ee

\begin{theorem}\label{th:H0positive}
Suppose that the operator $\bH_0$ is self-adjoint. Then:
\begin{itemize}
\item[(i)] $\bH_0$ is uniformly positive whenever $C(\cV)>0$.
\item[(ii)]  $\inf\sigma_{\ess}(\bH_0)>0$  if $C_{\ess}(\cV)>0$. 
\item[(iii)] The spectrum of $\bH_0$ is purely discrete if the number $\#\{e\in\cE\colon |e|>\varepsilon\}$ is finite for every $\varepsilon>0$ and $C_{\ess}(\cV)=\infty$.
\end{itemize}
\end{theorem}

\begin{proof}
Let $\varrho_0$ be a natural path metric on $\cV$ (see Lemma \ref{lem:varrho0}). Noting that $\varrho_0$ is an intrinsic metric on $\cV$, let us apply the Cheeger estimates from \cite{bkw15} to the discrete Laplacian $\rh_0$ given by \eqref{eq:h0}, \eqref{eq:mb}. First of all (see \cite[Section 2.3]{bkw15}), observe that the weighted area with respect to $\varrho_0$ is given by 
\[
{\rm Area}(\partial X) = \sum_{(u,v)\in\partial X} b(u,v)\varrho_0(u,v) = \sum_{(u,v)\in\partial X} \frac{1}{|e_{u,v}|}|e_{u,v}| = \sum_{(u,v)\in\partial X} 1 = \#(\partial X).
\]
Hence in this case the Cheeger estimate for discrete Laplacians (see Theorems 3.1 and 3.3 in \cite{bkw15}) implies the following estimates
\begin{align}\label{eq:Cheeger}
\inf \sigma(\rh_0) & \ge \frac{1}{2}C(\cV)^2, & \inf\sigma_{\ess}(\rh_0) & \ge \frac{1}{2} C_{\ess}(\cV)^2.
\end{align}
Combining these estimates with Corollary \ref{cor:kirchhoff}(ii)--(iii), we prove  (i) and (ii), respectively. 

Applying \cite[Theorem 3.3]{bkw15} once again, we see that the spectrum of $\rh_0$ is purely discrete if $C_{\ess}(\cV)=\infty$. Corollary \ref{cor:kirchhoff}(iv) finishes the proof of (iii). 
\end{proof}

Let $B_r(u)$ be a distance ball with respect to the natural path metric $\varrho_0$. Following \cite{hkw13} (see also \cite{kel15}), we define
\be
\mu:=\liminf_{r\to \infty} \frac{1}{r}\log m(B_r(v))
\ee
for a fixed $v\in\cV$, and 
\be
\underline{\mu}:=\liminf_{r\to \infty} \frac{1}{r}\inf_{v\in\cV}\log \frac{m(B_r(v))}{m(B_1(v))}.
\ee
Notice that $\mu$ does not depend on $v\in\cV$ if $\cV=\cup_{r\ge0} B_r(v)$.

\begin{theorem}\label{th:H0positive2}
Let $(\cV,\varrho_0)$ be complete as a metric space. Then:
\begin{itemize}
\item[(i)] $\inf\sigma(\bH_0)=0$ if $\underline{\mu}=0$. 
\end{itemize}
If in addition $m(\cV)=\infty$, then
\begin{itemize}
\item[(ii)] $\inf\sigma_{\ess}(\bH_0)=0$ if ${\mu}=0$. 
\item[(iii)] The spectrum of $\bH_0$ is not discrete if $\mu<\infty$. 
\end{itemize}
\end{theorem}

\begin{proof}
By Corollary \ref{cor:gaffney}, the operator $\bH_0$ is self-adjoint. The proof follows from the growth volume estimates on the spectrum of $\rh_0$. More precisely, the following bounds were established in \cite{hkw13} (see also \cite{fo,kel15}):
\begin{align*}
\inf \sigma(\rh_0) & \le \frac{1}{8}\underline{\mu}^2, & \inf \sigma_{\ess}(\rh_0) & \le \frac{1}{8}{\mu}^2.
\end{align*}
It remains to apply Corollary \ref{cor:kirchhoff}(ii)-(iv).
\end{proof}

We finish this section with a remark.

\begin{remark}\label{rem:Cheeger}
Connections between $\inf \sigma(\bH_0)$ and $\inf \sigma(\rh_0)$ and also between $\inf\sigma_{\ess}(\bH_0)$ and $\inf\sigma_{\ess}(\rh_0)$ by means of Theorem \ref{th:lsbTheta} and Theorem \ref{th:ess-impl-ess} are rather complicated since they  involve the corresponding Weyl function, which in our case has the form \eqref{eq:tiM}. In particular, it would be a rather complicated task to use these connections and then apply the Cheeger-type bounds for $\rh_0$ to estimate $\inf \sigma(\bH_0)$ and $\inf\sigma_{\ess}(\bH_0)$. For example, the following upper estimate, which easily follows from \eqref{eq:H0},
\[
\inf \sigma(\bH_0) \le  \inf\sigma(\bH^F) = \frac{\pi^2}{\sup_{e\in\cE}|e|^2}
\]
seems to be unrelated to $\inf \sigma(\rh_0)$. 
\end{remark}

\section{Spectral properties of quantum graphs with $\delta$-couplings}\label{sec:delta}

 In this section we are going to investigate spectral properties of the Hamiltonian $\bH_\alpha$ with $\delta$-couplings \eqref{eq:bcalpha} at the vertices. Namely, let $\alpha\colon\cV\to \R$ and the operator $\bH_\gA$ be defined in $L^2(\cG)$ as the closure of \eqref{eq:Halpha}. By Theorem \ref{th:main}, its spectral properties correlate with the corresponding properties of the discrete operator $\rh_\gA$ defined in $\ell^2(\cV;m)$ by \eqref{eq:h_a}.  In this section we shall always assume Hypotheses \ref{hyp:graph01} and \ref{hyp:graph02}.

\subsection{Self-adjointness and lower semiboundedness}\label{ss:5.1}

We begin with the study of the self-adjointness of the operator $\bH_\gA$.  
Our first result can be seen as a straightforward extension of Theorem \ref{th:saH0_1}.  

\begin{lemma}\label{lem:sa01}
If the weighted degree function $\Deg$ defined by \eqref{eq:wgtdeg} is bounded on $\cV$, that is, \eqref{eq:4.01} is satisfied, then the operator $\bH_\alpha$ is self-adjoint for any $\alpha\colon \cV\to \R$. Moreover, in this case the operator $\bH_\gA$ is bounded from below if and only if
\be\label{eq:Asemibd}
\inf_{v\in\cV} \frac{\gA(v)}{m(v)} >-\infty.
\ee
\end{lemma}

\begin{proof}
The operator of multiplication $A$ defined in $\ell^2(\cV,m)$ on the maximal domain $\dom(A) = \ell^2(\cV;\frac{\alpha^2}{m})$ by
\be\label{eq:Amultipl}
\begin{array}{cccc}
A\colon & \dom(A) & \to &\ell^2(\cV;m) \\
 & f & \mapsto & \frac{\alpha}{m}f
\end{array}
\ee
is clearly self-adjoint. If $\Deg$ is bounded on $\cV$, then the operator $\rh_0$ is bounded and self-adjoint in $\ell^2(\cV;m)$ (see \eqref{eq:estDavies}). It remains to note that $\rh_\alpha = \rh_0 + A$ and hence $\rh_\alpha$ is a self-adjoint operator since the self-adjointness is stable under bounded perturbations. Moreover, $\rh_\alpha$ is bounded from below if and only if so is $A$. The latter is clearly equivalent to \eqref{eq:Asemibd}. Theorem \ref{th:main}(i)-(ii) completes the proof.
\end{proof}

As an immediate corollary we arrive at the following result. 
 
\begin{corollary}\label{cor:sa02}
If $\inf_{e\in\cE} |e|>0$, then the operator $\bH_\alpha$ is self-adjont for any $\alpha\colon \cV\to \R$. Moreover, $\bH_\alpha$ is bounded from below if and only if $\alpha$ satisfies \eqref{eq:Asemibd}.
\end{corollary}

\begin{proof}
As in the proof of Corollary \ref{cor:sa01}, we get 
\[
 \sup_{v\in\cV} \Deg(v) \le \frac{1}{(\inf_{e\in\cE} |e|)^2}<\infty.
\]
It remains to apply Lemma \ref{lem:sa01}.
\end{proof}

\begin{remark}
A few remarks are in order.
\begin{itemize}
\item[(i)]
Using the form approach, the self-adjointness claim in Corollary \ref{cor:sa02} was proved in \cite[Section I.4.5]{bk13} under the additional assumption that $\frac{\alpha}{\deg}\colon\cV\to\R$ is bounded from below,  
\be\label{eq:BK}
\inf_{v\in\cV}\frac{\alpha(v)}{\deg(v)}>-\infty.
\ee
If $0<\inf_{e\in\cE}|e|\le \sup_{e\in\cE}|e|<\infty$, then it is easy to see that \eqref{eq:BK} is equivalent to \eqref{eq:Asemibd}.
\item[(ii)] 
Let us also mention that the graphs constructed in Examples \ref{ex:tree01} and \ref{ex:4.14} do not satisfy the condition of Corollary \ref{cor:sa02}, however, they satisfy \eqref{eq:4.01} and hence, by Lemma \ref{lem:sa01}, the corresponding Hamiltonian $H_\alpha$ is self-adjoint for any $\alpha\colon \cV\to \R$.
\end{itemize}
\end{remark}

The next result allows us to replace the boundedness assumption on the weighted degree by the local boundedness, however, now we need to assume some semiboundedness on $\alpha$. We begin with the following result.

\begin{proposition}\label{prop:sa03}
If the operator $\bH_0$ with Kirchhoff vertex conditions is self-adjoint in $L^2(\cG)$, then the operator $\bH_\gA$ with $\delta$-couplings on $\cV$ is self-adjoint whenever the function $\gA\colon \cV\to \R$ satisfies \eqref{eq:Asemibd}.
\end{proposition}

\begin{proof}
By Corollary \ref{cor:kirchhoff}(i), the discrete Laplacian $\rh_0$ given by \eqref{eq:h0}, \eqref{eq:mb} is a nonnegative self-adjoint operator in $\ell^2(\cV;m)$. On the other hand, \eqref{eq:Asemibd} implies that the multiplication operator $A$ defined by \eqref{eq:Amultipl} is a self-adjoint lower semibounded operator in $\ell^2(\cV;m)$. Noting that $C_c(\cV)$ is a core for both $\rh_0$ and $A$ since the graph is locally finite, we conclude that the operator $\rh_\alpha$ defined as a closure of the sum of $\rh_0$ and $A$ is a lowersemibounded self-adjoint operator in $\ell^2(\cV;m)$ (see \cite[Chapter VI.1.6]{kato}). It remains to apply Theorem \ref{th:main}(i).  
\end{proof}

\begin{remark}
 It follows from the proof of Proposition \ref{prop:sa03} and Theorem \ref{th:main}(ii) that the operator $\bH_\alpha$ is lower semibounded in this case.
\end{remark}

Combining Proposition \ref{prop:sa03} with the self-adjointness results from Section \ref{ss:4.2}, we can extend Corollary \ref{cor:sa02} to a much wider setting. Let us present only one result in this direction.

\begin{corollary} \label{lem:sa04}
Let $\varrho_m$ be the path metric \eqref{eq:p_m}, \eqref{eq:pathmetric} on $\cV$. If $(\cV,\varrho_m)$ is complete as a metric space and $\gA\colon\cV\to \R$ satisfies \eqref{eq:Asemibd}, then $\bH_\alpha$ is a lower semibounded self-adjoint operator. 

In particular, if the weight function $m$ satisfies \eqref{eq:4.12} and $\inf_{v\in\cV}\alpha(v)>-\infty$, then $\bH_\alpha$ is a lower semibounded self-adjoint operator.
\end{corollary}

\begin{proof}
Straightforward from Proposition \ref{prop:sa03}, Theorem \ref{prop:gaffney} and Corollary \ref{lem:sa02}. 
\end{proof}

\begin{remark}
Let us stress that both conditions (completeness of $(\cV,\varrho_m)$ and \eqref{eq:Asemibd}) are important. Indeed, 1-D Schr\"odinger operators with $\delta$-type interactions (see Example \ref{ex:01}) immediately provide counterexamples. First of all, in this setting completeness of $(\cV,\varrho_m)$ means that we consider a Schr\"odinger operator on an unbounded interval (either on the whole line $\R$ or on a semi-axis). Clearly, in the case of a compact interval the minimal operator is not self-adjoint even in the case of trivial couplings $\alpha\equiv 0$. On the other hand, it was proved in \cite{AKM_10} that in the case when all $\delta$-interactions are attractive ($\alpha_k<0$ for all $k\in\N$), the operator $H_\alpha$ given by \eqref{eq:1Ddelta} is bounded from below if and only if 
\be
\sup_{n\in\N} \sum_{x_k\in [n,n+1]} |\alpha_k| <\infty.
\ee
In the case $\inf_{k\in\N} (x_{k+1}-x_{k}) >0$ the latter is equivalent to $\inf_{k\in\N}\alpha_k>-\infty$.
\end{remark}

\subsection{Negative spectrum: CLR-type estimates}\label{ss:5.2}

Let $\alpha\colon \cV\to [0,\infty)$ be a nonnegative function on $\cV$. The main focus of this section is to obtain the estimates on the number of negative eigenvalues $\kappa_-(\bH_{-\alpha})$ of the operator $\bH_{-\alpha}$ in terms of the interactions $\alpha\colon \cV\to [0,\infty)$.  Note that by Theorem \ref{th:main}(iv), 
\be\label{eq:kappa-}
\kappa_-(\bH_{-\alpha}) = \kappa_-(\rh_{-\alpha}),
\ee
where $\rh_{-\alpha}$ is the (self-adjoint) discrete Laplacian defined either by \eqref{eq:discr_lapl} in $\ell^2(\cV)$ or by \eqref{eq:h_a} in $\ell^2(\cV;m)$. 

Suppose that the discrete Laplacian $\rh_0$ defined by \eqref{eq:h_a} with $\alpha\equiv 0$ is a self-adjoint operator in $\ell^2(\cV;m)$ (see Section \ref{ss:4.2}). 
It is well known (cf., e.g.,  \cite{fuk10}) that in this case $\rh_0$ generates a symmetric Markovian semigroup $\E^{-t \rh_0}$ (one can easily check that the Beurling--Deny conditions \cite{dav89, fuk10} are satisfied). Let us consider the corresponding quadratic form in $\ell^2(\cV;m)$:
\be\label{eq:Q0}
\gt_0[f]:= \frac{1}{2}\sum_{u,v\in\cV} b(v,u)|f(v) - f(u)|^2,\qquad f\in \dom(\gt_0):=\dom(\rh_0^{1/2}),
\ee
which is a regular Dirichlet form since $\cG$ is locally finite (see \cite{fuk10, kl12}). Recall that the functions $m$ and $b$ are given by \eqref{eq:mb}.

The following theorem is a particular case of \cite[Theorems 1.2--1.3]{ls97} (see also \cite[Theorem 2.1]{fls}). As it was already mentioned, $\rh_0$ generates a symmetric Markovian semigroup $\E^{-t \rh_0}$ in $\ell^2(\cV;m)$. Noting that $\rh_{-\alpha} f = \rh_0f - Af$ for all $f\in C_c(\cV)$, where $A$ is a multiplication operator  \eqref{eq:Amultipl},  and then applying \cite[Theorems 1.2--1.3]{ls97} (see also \cite[Theorem 2.1]{fls}) to the operator $\rh_{0}$, we arrive at the following result.

\begin{theorem}[\cite{ls97}]\label{th:leso}
Assume that $\rh_0$ is a self-adjoint operator in $\ell^2(\cV;m)$.
 Then the following conditions are equivalent:
\begin{itemize}
\item[(i)] There are constants $D>2$ and $K>0$ such that
\be\label{eq:Sob}
\|f\|^2_{\ell^{q}(\cV;m)}:=\left(\sum_{v\in\cV} |f(v)|^{q} m(v)\right)^{{2}/{q}} \le K \gt_0[f]
\ee
for all $f\in\dom(\gt_0)$ with $q=\frac{2D}{D-2}$. 
\item[(ii)]
There are constants $C>0$ and $D>2$ such that for all $\alpha\colon \cV\to [0,\infty)$ belonging to $\ell^{D/2}(\cV;m^{1-D/2})$ the form 
\[
\gt_{-\gA}[f] = \gt_0[f] - \sum_{v\in\cV}\gA(v)|f(v)|^2,\qquad \dom(\gt_{-\gA}):=\dom(\gt_0),
\] 
is bounded from below and closed in $\ell^2(\cV;m)$ and, moreover, 
the negative spectrum of $\rh_{- \alpha}$ is discrete and the following estimate holds
\be\label{eq:CLRabstr}
\kappa_-(\rh_{-\alpha}) \le C  \sum_{v\in\cV} \left(\frac{\alpha(v)}{m(v)}\right)^{D/2} m(v).
\ee
\end{itemize}
\end{theorem} 

\begin{remark}
\begin{itemize}
\item[(i)] The constants $K$ and $C$ in Theorem \ref{th:leso} are connected by $K^D\le C\le \E^{D-1} K^D$ (see \cite{fls}).
\item[(ii)] Since $C_c(\cV)$ is a core for both $\rh_0$ and $A$ whenever $\rh_0$ is essentially self-adjoint, it follows from Theorem \ref{th:leso} that the operator $\rh_{-\alpha}$ is bounded from below and self-adjoint for all $\alpha\in \ell^{D/2}(\cV;m^{1-D/2})$ if \eqref{eq:Sob} is satisfied.
\end{itemize}
\end{remark}

Combining Theorem \ref{th:main}(iv) with Theorem \ref{th:leso}, we immediately arrive at the following CLR-type estimate for quantum graphs with $\delta$-couplings at vertices.

\begin{theorem}\label{cor:clr01}
Assume that $\rh_0$ is a self-adjoint operator in $\ell^2(\cV;m)$. 
Then the following conditions are equivalent:
\begin{itemize}
\item[(i)] There are constants $D>2$ and $K>0$ such that \eqref{eq:Sob} holds 
for all $f\in\dom(\gt_0)$ with $q=\frac{2D}{D-2}$. 
\item[(ii)]
There are constants $C>0$ and $D>2$ such that for all $\alpha\colon \cV\to [0,\infty)$ belonging to $\ell^{D/2}(\cV;m^{1-D/2})$ the operator $\bH_{-\gA}$ is self-adjoint, bounded from below, its negative spectrum is discrete and the following estimate holds
\be
\kappa_-(\bH_{-\lambda \alpha}) \le C \lambda^{D/2} \sum_{v\in\cV} \left(\frac{\alpha(v)}{m(v)}\right)^{D/2} m(v),\quad \lambda>0.
\ee
\end{itemize}
The constants $K$ and $C$ are connected by $K^D\le C\le \E^{D-1} K^D$.
\end{theorem}

Of course, the most difficult part is to check the validity of the Sobolev-type inequality \eqref{eq:Sob}. However, there are several particular cases of interest when \eqref{eq:Sob} is known to be true (see \cite{gh}, \cite{sc}, \cite{var} and references therein).

\begin{corollary}\label{cor:clr02}
Let the metric graph $\cG=(\cV,\cE,|\cdot|)$ be such that the discrete graph $\cG_d=(\cV,\cE)$ is a Cayley graph of a group of polynomial growth $D$ with $D\ge 3$. 
If $\alpha\colon \cV\to [0,\infty)$ belongs to $\ell^{D/2}(\cV;m^{1-D/2})$, then 
\be
\kappa_-(\bH_{-\lambda \alpha}) \le C(\cG) \lambda^{D/2} \sum_{v\in\cV} \left(\frac{\alpha(v)}{m(v)}\right)^{D/2} m(v),\quad \lambda>0,
\ee
with some constant $C(\cG)$, which depends only on $\cG$.
\end{corollary}

\begin{proof}
By Theorem \ref{cor:clr01}, we only need to show that \eqref{eq:Sob} holds true. The argument is similar to \cite[Theorem 3.7]{ls97}. Indeed, by \cite[Theorem VI.5.2]{var}, since $\cG_d$ is a Cayley graph of the group of polynomial growth, there is a $C>0$ such that 
\be\label{eq:varop_group}
\|f\|_{\ell^q(\cV)} \le C \sum_{v\in \cV}\sum_{u\sim v} |f(v) - f(u)|^2, 
\ee
for all $f\in C_c(\cV)$ with $q = \frac{2D}{D-2}$. Since $ \sup_{e\in\cE}|e| <\infty$ (see Hypothesis \ref{hyp:graph02}), we get
\begin{align*}
\gt_0[f]=\frac{1}{2}\sum_{u,v\in\cV} b(v,u)|f(v) - f(u)|^2 =& \frac{1}{2}\sum_{v\in \cV}\sum_{u\sim v} \frac{1}{|e_{u,v}|}|f(v) - f(u)|^2 \\
&\ge \frac{1}{2\sup_{e\in\cE}|e|}\sum_{v\in \cV}\sum_{u\sim v} |f(v) - f(u)|^2,
\end{align*}
for all $f\in C_c(\cV)$. Combining this inequality with \eqref{eq:varop_group} and noting that 
\begin{align*}
\|f\|^q_{\ell^q(\cV;m)}=&\sum_{v\in\cV} |f(v)|^q m(v) = \sum_{v\in\cV} |f(v)|^q \sum_{e\in\cE_v}|e| \\
&\le \sup_{e\in\cE}|e| \sum_{v\in\cV} |f(v)|^q\deg(v) \le \|f\|^q_{\ell^q(\cV)} \sup_{e\in\cE}|e| \sup_{v\in\cV}\deg(v),
\end{align*}
we get \eqref{eq:Sob}. 
\end{proof}

\begin{remark}
Notice that in Corollary \ref{cor:clr02} we did not make any additional assumptions on the weight function $m$. Namely, we only assumed that the edges lengths satisfy \eqref{eq:dstar}.
\end{remark}

In particular, in the case $\cG_d = \Z^N$ we get the following estimate.
\begin{corollary}\label{eq:clr03}
Let $\cG_d=\Z^N$ with $N\ge 3$. Also, assume that \eqref{eq:dstar} is satisfied. If $\alpha\colon \cV\to [0,\infty)$ belongs to $\ell^{\frac{N}{2}}(\Z^N;m^{1-N/2})$, then 
\be
\kappa_-(\bH_{-\lambda \alpha}) \le C_N \lambda^{N/2}  \sum_{v\in\cV} \left(\frac{\alpha(v)}{m(v)}\right)^{N/2} m(v),\quad \lambda>0,
\ee
with some constant $C_N$, which depends only on $N$ and $m$.
\end{corollary}

It was first noticed by G.\ Rozenblum and M.\ Solomyak (see \cite[Theorem 3.1]{roso09} and also \cite{roso10}) that in contrast to Schr\"odinger operators on $\R^N$, in the case $\cG_d=\Z^N$ for every $q\in (0,D/2)$ the following holds 
\be
\kappa_-(\rh_{-\lambda\alpha}) = \OO(\lambda^{q}),\quad \lambda\to +\infty,
\ee
whenever $\inf_{e\in\cE}|e|>0$ and $\alpha\in \ell^q_w(\cV)$, that is, 
\[
\#\{v\in\cV\colon |\alpha(v)|>n\} = \OO(n^{-q})
\]
as $n\to \infty$ or equivalently $\ti{\alpha}_n = \OO(n^{-1/q})$ as $n\to \infty$, where $\{\ti{\alpha}_n\}_{n\in\N}$ is a re-arrangement of $\{\alpha(v)\}_{v\in\cV}$ in a decreasing order. Define
\[
\|\alpha\|_{\ell^q_w}:= \sup_n n^{1/q} \ti{\alpha}_n.
\]
It turns out that the later holds in a wider setting and hence we arrive at the following result.

\begin{proposition}\label{th:6.5}
Assume the conditions of Theorem \ref{cor:clr01}. If $\cG$ satisfies \eqref{eq:4.01}, then for every $q\in (0,D/2)$
\be\label{eq:6.10}
\kappa_-(\bH_{-\lambda\alpha}) \le  C\lambda^{q} \|\alpha\|_{\ell^q_w}^q,\quad \lambda>0,
\ee
whenever $\alpha\in \ell^q_w(\cV)$. Here the constant $C$ depends only on $q$, $D$ and $\cV$.
\end{proposition}

\begin{proof}
By Theorem \ref{th:main}(iv), we only need to show that 
\be\label{eq:6.11}
\kappa_-(\rh_{-\lambda\alpha}) \le  C\lambda^{q} \|\alpha\|_{\ell^q_w}^q,\quad \lambda>0.
\ee
The validity of \eqref{eq:6.11} was established in \cite[Theorem 3.1]{roso10} under the additional assumptions $\inf_{e\in\cE}|e|>0$ and $\sup_{v\in\cV}\deg(v)<\infty$. In fact, this proof (see also \cite[\S 3]{roso09}) can be extended line by line to the case of graphs $\cG$ satisfying \eqref{eq:4.01}.
\end{proof}

\begin{remark}
For a further discussion of eigenvalue estimates for discrete operators and quantum graphs on the lattice $\Z^N$ we refer to \cite{roso11}.
\end{remark}

\begin{remark}\label{rem:5.14}
To a large extent, the behavior of the negative spectrum of $\rh_{- \alpha}$ is determined by the behavior of the following function 
\be
g(t):= \sup_{u,v\in\cV} |P(t;u,v)| = \|\E^{-t \rh_0}\|_{\ell^1\to \ell^\infty},
\ee
where $P(t;\cdot,\cdot):=\E^{-t \rh_0}(\cdot,\cdot)$ is the heat kernel (see \cite{roso97, roso10} and also \cite{fls, mova10, mova12}). 
In particular, the exponents $d$ and $D$ determined by
\begin{align}\label{eq:dimensions}
g(t) & = \OO(t^{-d/2}),\quad t\to 0, & g(t) & = \OO(t^{-D/2}),\quad t\to +\infty,
\end{align}
and called the local dimension and the global dimension, respectively, are very important in the analysis of $\kappa_-(\rh_{-\alpha})$ (see Section 2 in \cite{roso97}). By \cite[Theorem II.5.2]{var}, \eqref{eq:Sob} is equivalent to the following estimate 
\be\label{eq:decay}
g(t) \le C t^{-D/2},\quad t>0,
\ee
with some positive constant $C>0$. On the other hand,  $d=0$ if \eqref{eq:4.01} holds, that is, if $\rh_0$ is a bounded operator and, moreover, $\ell^1(\cV)\subset \ell^2(\cV)\subset \ell^\infty(\cV)$. It is precisely this fact which allows to prove Proposition \ref{th:6.5}. Note that $d=D=N$ for Schr\"odinger operators on $\R^N$ and hence the estimates of the type \eqref{eq:6.10} have no analogues in this case. 
\end{remark}

Equality \eqref{eq:kappa-} together with Remark \ref{rem:5.14} indicate that there is a close connection between the heat semigroups $\E^{-t\rh_0}$ and $\E^{-t\bH_0}$. In fact, the following holds true.

\begin{theorem}\label{lem:rozsol}
Assume that $\rh_0$ and $\bH_0$ are self-adjoint operators in $\ell^2(\cV;m)$ and $L^2(\cG)$, respectively. Then the following statements are equivalent
\begin{itemize}
\item[(i)] $\|\E^{-t\rh_0}\|_{\ell^1\to \ell^\infty} \le  C_1t^{-D/2}$ holds for all $t>0$ with some $C_1>0$ and $D>2$,
\item[(ii)]  $\|\E^{-t\bH_0}\|_{L^1\to L^\infty} \le  C_2t^{-D/2}$ holds for all $t>0$ with some $C_2>0$ and $D>2$.
\end{itemize}
Here the constants $C_1$ and $C_2$ might be different.
\end{theorem}

\begin{proof}
By Varopoulos's theorem (see \cite[Theorem II.5.2]{var}), (i) and (ii) are equivalent to the validity of the corresponding Sobolev type inequalities. Namely, (i) is equivalent to \eqref{eq:Sob} and (ii) is equivalent to the inequality
\be\label{eq:Sob_cont}
\left(\int_{\cG} |f(x)|^q dx\right)^{2/q} \le C \int_{\cG} |f'(x)|^2 dx,\qquad f\in H^1(\cG),
\ee
where $H^1(\cG)$ is the Sobolev space on $\cG$, which coincides with the form domain of the operator $\bH_0$, and $q=\frac{2D}{D-2}$ and $D>2$. Hence it suffices to show that \eqref{eq:Sob} is equivalent to \eqref{eq:Sob_cont}.

First observe that every $f\in H^1(\cG)$ admits a unique decomposition $f=f_{\rm lin} + f_0$, where $f_{\rm lin}\in H^1(\cG)$ is piecewise linear on $\cG$ and $f_0\in H^1(\cG)$ takes zero values at the vertices $\cV$. It is easy to check that 
\[
\gt_{\bH_0}[f]=\int_{\cG} |f'(x)|^2 dx = \int_{\cG} |f_{\rm lin}'(x)|^2 dx + \int_{\cG} |f_0'(x)|^2 dx = \gt_{\bH_0}[f_{\rm lin}] + \gt_{\bH_0}[f_0]. 
\]
Moreover, we have (see Remark \ref{rem:3.8}):
\[
 \gt_{\bH_0}[f_{\rm lin}] = \gt_{\rh_0}[f_{\rm lin}],\qquad f_{\rm lin}\in H^1(\cG)\cap \cL.
\]

Next it is easy to see that \eqref{eq:Sob_cont} holds for all $f=f_0\in H^1(\cG)$ with $q>2$ and with a constant $C(\cG)$ which depends only on $\sup_{e\in\cE}|e|$ and $q>2$. Noting that every piecewise linear function $f=f_{\rm lin}\in H^1(\cG)\cap \cL$ satisfies
\begin{align*}
\|f\|^q_{L^q(\cG)} = & \sum_{e\in\cE}\int_{e} |f(x)|^q dx \le \sum_{e\in\cE} |e| \max_{x\in e}|f(x)|^q \\
&\le \sum_{e\in\cE} |e|(|f_e(e_i)|^q + |f_e(e_o)|^q) = 2\sum_{v\in\cV}|f(v)|^qm(v) = 2\|f\|^q_{\ell^q(\cV;m)},
\end{align*}
we conclude that (i) implies (ii).

Clearly, to prove that (ii) implies (i) it suffices to show that every linear function $f$ on a finite interval $(a,b)$ satisfies the estimate
\be\label{eq:lin}
(|f(a)|^q +|f(b)|^q) \le \frac{C}{b-a} \int_a^b|f(x)|^q dx,
\ee
where $C>0$ is a positive constant which depends only on $q>2$. Indeed, we have (cf. Remark \ref{rem:3.8})
\be\label{eq:lin01}
\int_a^b |f(x)|^2dx = (b-a) \frac{|f(a)|^2 + \re(f(a)f(b)) + |f(b)|^2}{3}.
\ee
Applying the H\"older inequality to the left-hand side in \eqref{eq:lin01}, one gets 
\be\label{eq:lin02}
\int_a^b |f(x)|^2dx \le (b-a)^{1/p}\left(\int_a^b |f(x)|^qdx\right)^{2/q},\qquad \frac{1}{p}=1-\frac{2}{q}.
\ee
On the other hand, applying the Cauchy--Schwarz inequality to the right-hand side in \eqref{eq:lin01}, we arrive at 
\[ 
\frac{|f(a)|^2 + \re(f(a)f(b)) + |f(b)|^2}{3} \ge \frac{|f(a)|^2 + |f(b)|^2}{6} \ge \frac{(|f(a)|^q + |f(b)|^q)^{2/q}}{6c(q)},
\]
where $c(q)>0$ depends only on $q>2$.
 Combining this estimate with  \eqref{eq:lin01} and \eqref{eq:lin02}, we obtain \eqref{eq:lin}, which implies that 
\[
(6c(q))^{-q/2}\|f\|^q_{\ell^q(\cV;m)} \le \|f\|^q_{L^q(\cG)} 
\]
holds for all $f=f_{\rm lin}\in H^1(\cG)\cap \cL$.
\end{proof}

\begin{remark}\label{rem:rozsol}
The implication $(i) \Rightarrow  (ii)$ in Theorem \ref{lem:rozsol} was observed by Rozenblum and Solomyak (see \cite[Theorem 4.1]{roso10}),
however, for a different discrete Laplacian defined by \eqref{eq:1.03B}, where the weight function $m\colon v\mapsto \sum_{e\in\cE_v}|e|$ is replaced by the vertex degree function $\deg\colon v\mapsto \#(\cE_v)$.  Since 
\[
m(v) \le \deg(v) \sup_{e\in\cE}|e|
\]
for all $v\in \cV$ and $\sup_{e\in\cE}|e|<\infty$, $\ell^2(\cV;\deg)$ is continuously embedded into $\ell^2(\cV;m)$, however, the converse is not true. This together with Theorem~\ref{th:leso} imply that one cannot replace \eqref{eq:h0} by \eqref{eq:1.03B} in Theorem \ref{lem:rozsol} and the converse statement to Theorem~4.1 in \cite{roso10} is not true without further assumptions on the function $m$.
\end{remark}

\subsection{Spectral types}\label{ss:5.3}
In this subsection we plan to investigate the structure of the spectrum of $\bH_\alpha$.

\subsubsection{Resolvent comparability} \label{ss:5.3.1}

We begin with the following simple corollary of Theorem \ref{th:main}(viii).

\begin{corollary}\label{cor:5.01}
Assume the conditions of Theorem \ref{th:main}. 
\begin{itemize}
\item[(i)] If $\frac{\alpha-\ti{\alpha}}{m}\in c_0(\cV)$, then $\sigma_{\ess}(\bH_\alpha) = \sigma_{\ess}(\bH_{\wt\alpha})$. In particular, if $\frac{\alpha}{m} \in c_0(\cV)$, then $\sigma_{\ess}(\bH_\alpha) = \sigma_{\ess}(\bH_0)$. 
\item[(ii)] If $\frac{\alpha-\ti{\alpha}}{m}\in \ell^1(\cV)$, then $\sigma_{\ac}(\bH_\alpha) = \sigma_{\ac}(\bH_{\wt\alpha})$. In particular, if $\frac{\alpha}{m} \in \ell^1(\cV)$, then $\sigma_{\ac}(\bH_\alpha) = \sigma_{\ac}(\bH_0)$. 
\end{itemize}
\end{corollary}

Here $\alpha\in c_0(\cV)$ means that the set $\{v\in\cV\colon |\alpha(v)|>\varepsilon\}$ is finite for every $\varepsilon>0$.

\begin{proof}
It suffices to note that $\rh_{\alpha} f - \rh_{\ti\alpha} f = \frac{\alpha-\ti{\alpha}}{m} f$ for all $f\in C_c(\cV)$. Hence $(\rh_\alpha - \I)^{-1} - (\rh_{\ti\alpha} - \I)^{-1} \in \gS_\infty$ if  $\frac{\alpha-\ti{\alpha}}{m}\in c_0(\cV)$ and then, by the Weyl theorem and Theorem \ref{th:main}(viii), we prove the first claim.

Moreover, $(\rh_\alpha - \I)^{-1} - (\rh_{\ti\alpha} - \I)^{-1} \in \gS_1$ whenever  $\frac{\alpha-\ti{\alpha}}{m}\in \ell^1(\cV)$. It remains to apply Theorem \ref{th:main}(viii) and the Birman--Krein theorem.
\end{proof}

The presence (or absence) of an absolutely continuous spectrum for quantum graphs $\bH_0$ with Kirchhoff vertex conditions at vertices is a challenging open problem. To the best of our knowledge, some partial results have been obtained in the cases of radially symmetric trees and for some special classes of (equilateral) graphs that originate from groups, e.g., the corresponding Cayley graphs or Schreier graphs (see, e.g., \cite{bf07}, \cite{ekkst08}, \cite{exli}, \cite{ess13}, \cite{sol02}). In particular, it is shown in \cite[Theorem 5.1]{ess13} that in the case when $\cG$ is a rooted radial tree with a finite complexity of the geometry, the absolutely continuous spectrum of $\bH_0$ is nonempty if and only if $\cG$ is eventually periodic. 

Our next result provides a sufficient condition for $\bH_\alpha$ to have purely singular spectrum.

\begin{theorem}\label{th:stolz}
Assume that $\inf_{e\in\cE}|e|>0$ and $\sup_{e\in\cE}|e|<\infty$. If $\alpha\colon \cV\to \R$ is such that for any infinite path $\cP\subset \cG$ without cycles
\be\label{eq:stolz}
\sup_{v\in \cP} \frac{|\alpha(v)|}{\deg(v)} = \infty,
\ee  
then $\sigma_{\ac}(\bH_\alpha)=\emptyset$.
\end{theorem}

\begin{proof}
The proof is based on the standard trace class argument \cite{ss89}. By Corollary \ref{cor:sa02}, the operator $\bH_\alpha$ is self-adjoint. Since \eqref{eq:stolz} holds for every infinite path $\cP\subset \cG$, we can find a subset $\wt{\cV}\subset \cV$ such that 
\be\label{eq:6.2}
\sum_{v\in\wt{\cV}}\frac{\deg(v)}{|\alpha(v)|} <\infty
\ee
and the graph $\cG$ is a countable union of finite subgraphs $\cG_k$, $k\in\N$ such that the boundary $\partial \cG_k$ of every subgraph $\cG_k$ is contained in $\wt{\cV}$. Define a new function $\tilde{\alpha}\colon\cV\to \R\cup \{\infty\}$ by 
\be
\tilde{\alpha}(v) = \begin{cases} \alpha(v), & v\in \cV\setminus \wt{\cV},\\ \infty, & v\in\wt{\cV},\end{cases}
\ee
that is, at every vertex $v\in\cV\setminus\wt{\cV}$ the corresponding boundary condition for $\bH_{\tilde{\alpha}}$ is given by \eqref{eq:bcalpha} and at every vertex $v\in\wt{\cV}$ it has the Dirichlet boundary condition. Let us show that 
\be\label{eq:6.4}
(\bH_\alpha - \I)^{-1} - (\bH_{\tilde{\alpha}} - \I)^{-1} \in \gS_1.
\ee

It is easy to see that under the assumptions $\inf_{e\in\cE}|e|>0$ and $\sup_{e\in\cE}|e|<\infty$ the triplet $\wt{\Pi} = \{\cH_{\cG},\wt{\Gamma}_0^0,\wt{\Gamma}_1^0\}$ 
given by \eqref{eq:wtG0}, \eqref{eq:wtG1} is a boundary triplet for $\bH_{\max}$. Next we set
\begin{align}
C_\alpha & := \bigoplus_{v\in\cV} C_{v,\alpha}, & D_\alpha & := \bigoplus_{v\in\cV} D_{v},
\end{align}
where $C_{v,\alpha}$ and $D_v$ are given by \eqref{eq:3.7}, and 
\begin{align}
\wt{C}_{\tilde{\alpha}} & := \bigoplus_{v\in\cV} \wt{C}_{v,\tilde{\alpha}}, & \wt{D}_{\tilde{\alpha}} & := \bigoplus_{v\in\cV} \tilde{D}_{v},
\end{align}
where
\begin{align}
\wt{C}_{v,\tilde{\alpha}} & = \begin{cases} C_{v,\alpha}, & v\in \cV\setminus\wt{\cV} \\ I_{\deg(v)}, & v\in\wt{\cV} \end{cases}, &
 \wt{D}_{v} & = \begin{cases} D_{v}, & v\in \cV\setminus\wt{\cV} \\ \bO_{\deg(v)}, & v\in\wt{\cV} \end{cases}.
\end{align}
Observe that the corresponding boundary relations $\Theta_\alpha$ and $\Theta_{\tilde{\alpha}}$ parameterizing $\bH_\alpha$ and $\bH_{\tilde{\alpha}}$ via the boundary triplet $\Pi_\cG = \{\cH_{\cG},\wt{\Gamma}_0^0,\wt{\Gamma}_1^0\}$ are the closures of
\[
\Theta_\alpha^0=\{\{f,g\}\in \cH_{\cG}\times\cH_{\cG}\colon C_\alpha f=D_\alpha g\},\quad \Theta_{\tilde{\alpha}}^0=\{\{f,g\}\in \cH_{\cG}\times\cH_{\cG}\colon \wt{C}_{\tilde{\alpha}} f=\wt{D}_{\tilde{\alpha}} g\}.
\] 
Straightforward calculations show that  
\[
\tr\left((\Theta_\alpha -\I)^{-1} - (\Theta_{\tilde{\alpha}} - \I)^{-1}\right) = \sum_{v\in\wt{\cV}} \Big(\frac{\alpha(v)}{\deg(v)} - \I\Big)^{-1},
\]
which is finite according to \eqref{eq:6.2}. Therefore, by Theorem \ref{th:properext}(iv), \eqref{eq:6.4} holds true. It remains to note that $\bH_{\tilde{\alpha}}$ is the orthogonal sum of operators having discrete spectra and hence the spectrum of $\bH_{\tilde{\alpha}}$ is pure point. The Birman--Krein theorem then yields $\sigma_{\ac}(\bH_\alpha) = \sigma_{\ac}(\bH_{\tilde{\alpha}}) = \emptyset$.
\end{proof}

\begin{corollary}\label{cor:sstree}
Let $\cG$ be a rooted radially symmetric tree with the root $o$ and such that $\inf_{e\in\cE}|e|>0$ and $\sup_{e\in\cE}|e|<\infty$. Also, let  $\alpha\colon\cV\to \R$ be radially symmetric, that is,   
$\alpha(v) = \alpha_k$ for all $v\in\cV$ such that $d(o,v) = k$, where$d(o,v)$ is the combinatorial distance from $v$ to the root $o$. If 
\be
\sup_{k\in\N} \frac{|\alpha_k|}{\deg(v_k)} = \infty,
\ee
then $\sigma_{\ac}(\bH_\alpha)=\emptyset$. 
\end{corollary}

\begin{remark}
Corollary \ref{cor:sstree} can be seen as the analog of \cite[Theorem 3]{ShuSto94} and \cite[Theorem 1]{mik96}. Moreover, the assumption $\sup_{e\in\cE}|e|<\infty$ in Theorem \ref{th:stolz} and Corollary \ref{cor:sstree} can be removed by adding inessential vertices.
\end{remark}

\subsubsection{Bounds on the spectrum of $\bH_\alpha$}\label{ss:5.3.2}

Throughout this subsection we shall assume that $\alpha\colon \cV\to [0,\infty)$, that is, all interactions at vertices are nonnegative. Let $\varrho$ be an intrinsic metric. In order to include $\alpha$ into Cheeger type estimates, we need to modify the definition of Cheeger constants \eqref{eq:isoper} and \eqref{eq:isoperess} following \cite{kl10}, \cite{bkw15}. 
For every subgraph $\wt{\cV}\subseteq \cV$ one defines {\em the modified isoperimetric constant} 
\be\label{eq:isoperA}
C_\alpha(\wt{\cV}) := \inf_{X\subset \wt{\cV}} \frac{{\rm Area}_\alpha(\partial X)}{m(X)},
\ee
where 
\be
{\rm Area}_\alpha(\partial X) := \sum_{(u,v)\in\partial X} b(u,v)\varrho_0(u,v) + \sum_{v\in X}\gA(v)= \sum_{(u,v)\in\partial X} 1 + \sum_{v\in X}\gA(v),
\ee
and
\be
 m(X) = \sum_{v\in X} m(v).
\ee
Moreover, we need {\em the isoperimetric constant at infinity}
\be\label{eq:isoperessA}
C_{\ess,\gA}(\cV) := \sup_{X\subset \cV\ \text{is finite}} C_\gA(\cV\setminus X).
\ee

\begin{theorem}\label{th:Hapositive}
Suppose that the operator $\bH_\gA$ is self-adjoint. Then:
\begin{itemize}
\item[(i)] $\bH_\gA$ is uniformly positive if $C_\gA(\cV)>0$.
\item[(ii)] $ \inf\sigma_{\ess}(\bH_\gA)>0$ if $C_{\ess,\gA}(\cV)>0$. 
\item[(iii)] The spectrum of $\bH_\gA$ is discrete if the number $\#\{e\in\cE\colon |e|>\varepsilon\}$ is finite for every $\varepsilon>0$ and $C_{\ess,\gA}(\cV)=\infty$.
\end{itemize}
\end{theorem}

\begin{proof}The proof is analogous to that of Theorem \ref{th:H0positive} and we only need to use the corresponding modifications of Cheeger type bounds for the discrete operator $\rh_\gA$ from \cite{bkw15}.
\end{proof}

\section{Other boundary conditions}\label{sec:final}

In the present paper our main focus was on the Kirchhoff and $\delta$-type couplings at vertices (see \eqref{eq:bcalpha}). There are several other physically relevant classes of couplings (see, e.g., \cite{bk13,chex,ex96}). Our main result, Theorem \ref{th:HTheta}, covers all possible cases, however, the key problem is to calculate the boundary operator and then to investigate its spectral properties. It turned out that for $\delta$-couplings the corresponding boundary operator is given by the discrete Laplacian \eqref{eq:h_a}, which attracted an enormous attention during the last three decades. However, for other boundary conditions new nontrivial discrete operators of higher order may arise. For example, this happens in the case of the so-called {\em $\delta_s'$-couplings}, cf. \cite{ex96}. Namely (see \cite{chex,ex96}), let $\beta\colon \cV\to \R$ and consider the following boundary conditions at the vertices $v\in\cV$:
\be\label{eq:bcbeta}
\begin{cases} \frac{df}{dx_e}(v) \ \text{does not depend on $e$ at the vertex}\ v,\\[2mm] \sum_{e\in \cE_v}f(v) = \beta(v)\frac{df}{dx_e}(v). \end{cases}
\ee
Define the corresponding operator $\bH_{\beta}$ as the closure of the operator $\bH_{\beta}^0$ given by
\begin{align}
\bH_{\beta}^0 &= \bH_{\max}\upharpoonright {\dom(\bH_{\beta}^0)},\nn\\ 
&\dom(\bH_{\beta}^0) = \big\{f\in \dom(\bH_{\max})\cap L^2_{c}(\cG)\colon f\ \text{satisfies}\ \eqref{eq:bcbeta},\ v\in\cV\big\}.\label{eq:Hbeta}
\end{align}

To avoid lengthy and cumbersome calculations of the corresponding boundary relation $\Theta_\beta$ parameterizing $\bH_\beta$ with the help of the boundary triplet $\Pi$ constructed in Corollary \ref{cor:2.4}, let us consider the kernel $\cL= \ker(\bH_{\max})$ of $\bH_{\max}$ as in Remark \ref{rem:3.8}. Recall that $\cL= \ker(\bH_{\max})$ consists of piecewise linear functions on $\cG$ and every $f\in \cL$ can be identified with its values on $\cV$, $\{f(e_i), f(e_o)\}_{e\in \cE}$. Moreover, the $L^2$ norm of  $f\in \cL$
is equivalent to 
\[
\sum_{e\in\cE} |e| ({|f(e_i)|^2 + |f(e_o)|^2}).
\]
It is not difficult to see that  (see also \cite[p.27]{bk13})
\[
(\bH_\beta f,f) = \sum_{e\in\cE} \int_{e} |f'({\rm x})|^2 d{\rm x} + \sum_{v\in\cV} \frac{1}{\beta(v)}\Big|\sum_{e\in \cE_v}f_e(v)\Big|^2,\quad f\in\cL\cap L^2_c(\cG).
\]
Therefore, for every $f\in \cL\cap L^2_c(\cG)$ we get 
\be\label{eq:7.5}
(\bH_\beta f,f) = \sum_{e\in\cE} \frac{|f(e_o) - f(e_i)|^2}{|e| } + \sum_{v\in\cV} \frac{1}{\beta(v)}\Big|\sum_{e\in \cE_v}f_e(v)\Big|^2.
\ee
Clearly, the right-hand side in \eqref{eq:7.5} is a form sum of two difference operators, where the first one is the standard discrete Laplacian, however, the second one gives rise to a difference expression of higher order. In particular, its order at every vertex equals the degree $\deg(v)$ of the corresponding vertex $v\in\cV$. Unfortunately, we are not aware of the literature where the difference operators of this type have been studied. 

\appendix

\section{Boundary triplets and Weyl functions}\label{app:triplets}

\subsection{Linear relations}\label{ss:a1}

Let $\cH$ be a separable Hilbert space. {\em A (closed) linear relation} in $\cH$ is a (closed) linear subspace in $\cH\times\cH$. The set of all closed linear relations is denoted by $\wt{\cC}(\cH)$. Since every linear operator in $\cH$ can be identified with its graph, the set of linear operators can be seen as a subset of all linear relations in $\cH$. In particular, the set of closed linear operators $\cC(\cH)$ is a subset of $\wt{\cC}(\cH)$.

Recall that the domain, the range, the kernel and the multivalued part of a linear relation $\Theta$ are given, respectively, by
\begin{align*}
\dom(\Theta) &= \{f\in \cH\colon \exists g\in\cH\ \text{such that}\ \{f,g\}\in \Theta\}, \\
\ran(\Theta) &= \{g\in \cH\colon  \exists f\in\cH\ \text{such that}\   \{f,g\}\in \Theta\},\\ 
\ker(\Theta) &= \{f\in \cH\colon  \{f,0\}\in \Theta\},\\
\mul(\Theta) &= \{g\in \cH\colon  \{0,g\}\in \Theta\}.
\end{align*} 

The adjoint linear relation $\Theta^\ast$ is defined by
\be
\Theta^\ast = \big\{ \{\ti{f},\ti{g}\}\in \cH\times\cH\colon  (g,\ti{f})_{\cH} = (f,\ti{g})_{\cH}\ \text{for all}\ \{f,g\}\in\Theta\big\}.
\ee
$\Theta$ is called {\em symmetric} if $\Theta\subset \Theta^\ast$. If $\Theta=\Theta^\ast$, then it is called {\em self-adjoint}. Note that $\mul(\Theta)$ is orthogonal to $\dom(\Theta)$ if $\Theta$ is symmetric. Setting $\cH_{\op}:=\overline{\dom(\Theta)}$, we obtain the orthogonal decomposition of a symmetric linear relation $\Theta$:
\be\label{eq:ThetaDecomp}
\Theta = \Theta_{\op}\oplus \Theta_{\infty},
\ee 
where $\Theta_\infty = \{0\}\times \mul(\Theta)$ and $\Theta_{\op}$ is a symmetric linear operator in $\cH_{\op}$, called the {\em operator part} of $\Theta$. 

The inverse of the linear relation $\Theta$ is given by
\[
\Theta^{-1} = \{\{g,f\}\in \cH\times\cH\ \colon\ \{f,g\}\in \Theta\}.
\] 
The sum of linear relations $\Theta_1$ and $\Theta_2$ is defined by
\[
\Theta_1+\Theta_2 = \{\{f,g_1+g_2\}:\ \{f,g_1\}\in\Theta_1,  \ \{f,g_2\}\in\Theta_2\}.
\] 
Hence one can introduce the resolvent $(\Theta - z)^{-1}$ of the linear relation $\Theta$, which is well defined for all $z\in\C$. However, the set of those $z\in\C$ for which $(\Theta - z)^{-1}$ is a graph of a closed bounded operator in $\cH$ is called the {\em resolvent set} of $\Theta$ and is denoted by $\rho(\Theta)$. Its complement 
$\sigma(\Theta)=\C\setminus\rho(\Theta)$ is called the {\em spectrum} of $\Theta$. If $\Theta$ is symmetric, then taking into account \eqref{eq:ThetaDecomp} we obtain
\be\label{eq:ThetaResolv}
(\Theta - z)^{-1} = (\Theta_{\op} - z)^{-1}\oplus \bO_{\mul(\Theta)}.
\ee 
This immediately implies that $\rho(\Theta) = \rho(\Theta_{\op})$, $\sigma(\Theta) = \sigma(\Theta_{\op})$ and, moreover, one can introduce the spectral types of $\Theta$ as those of its operator part $\Theta_{\op}$.

Let us mention that self-adjoint linear relations admit a very convenient representation, which was first obtained by Rofe-Beketov \cite{RB} in the finite dimensional case (see also \cite[Exercises 14.9.3-4]{schm}).

\begin{proposition}\label{prop:RB}
Let $C$ and $D$ be bounded operators on $\cH$ and 
\be\label{eq:ThetaCD}
\Theta_{C,D}:= \big\{ \{f,g\}\in \cH\times\cH\colon Cf=Dg \big\}.
\ee
Then $\Theta_{C,D}$ is self-adjoint if and only if 
\begin{align}\label{eq:condRB}
CD^\ast & =DC^\ast, & \ker\begin{pmatrix} C& -D \\ D & C\end{pmatrix} & =\{0\}.
\end{align}
If $\dim \cH =N<\infty$, then the second condition in \eqref{eq:condRB} is equivalent to ${\rm rank}(C|D)=N$.
\end{proposition}

Further details and facts about linear relations in Hilbert spaces can be found in, e.g., \cite[Chapter 6.1]{DM17}, \cite[Chapter 14]{schm}.

\subsection{Boundary triplets and proper extensions}\label{ss:a2}

Let $A$ be a densely defined closed symmetric operator in a
separable Hilbert space $\gH$ with equal deficiency indices
$\mathrm{n}_\pm(A)=\dim \cN_{\pm \I} \leq \infty,\ \ \cN_z:=\ker(A^*-z)$.

\begin{definition}[\cite{Gor84}]\label{def_ordinary_bt}
A triplet $\Pi=\{\cH,\gG_0,\gG_1\}$ is called {\em a  boundary
triplet} for the adjoint operator $A^*$ if $\cH$ is a Hilbert
space and $\Gamma_0,\Gamma_1\colon  \dom(A^*)\rightarrow \cH$ are
bounded linear mappings such that the abstract Green's identity
\begin{equation}\label{eq:green_f}
(A^*f,g)_\gH - (f,A^*g)_\gH = (\gG_1f,\gG_0g)_\cH - (\gG_0f,\gG_1g)_\cH
\end{equation}
holds for all $f,g\in\dom(A^*)$  and the mapping 
\be\label{eq:BTGamma}
\begin{array}{cccc}
\gG\colon &  \dom(A^*) & \to & \cH\times\cH \\
                & f    &           \mapsto & \{\Gamma_0f,\Gamma_1 f\}
\end{array}
\ee
 is surjective.
\end{definition}

A boundary triplet for $A^*$ exists if and only if the deficiency indices of $A$ are equal (see, e.g.,  \cite[Prop.7.4]{DM17}, \cite[Prop. 14.5]{schm}). Moreover, $\mathrm{n}_\pm(A) = \dim(\cH)$ and $A=A^*\upharpoonright \ker(\Gamma)$. Note also that the boundary triplet for $A^*$ is not unique. 

An extension $\ti{A}$ of $A$ is called {\em proper} if $\dom(A)\subset \dom(\ti{A}) \subset \dom(A^\ast)$. The set of all proper extensions is denoted by $\Ext(A)$. 

\begin{theorem}[\cite{DM91,mmm92b}]\label{th:properext}
Let $\Pi = \{\cH,\Gamma_0,\Gamma_1\}$ be a boundary triplet for $A^\ast$. Then the mapping $\Gamma$ defines a bijective correspondence between $\Ext(A)$ and the set of all linear relations in $\cH$:
\be\label{eq:ATheta}
\Theta \mapsto A_\Theta:= A^\ast\upharpoonright {\{f\in\dom(A^\ast)\colon \ \Gamma f=\{\Gamma_0f,\Gamma_1f\}\in \Theta \}}.
\ee
Moreover, the following holds:
\begin{itemize}
\item[(i)] $A_\Theta^\ast = A_{\Theta^\ast}$.
\item[(ii)] $A_\Theta\in \cC(\gH)$ if and only if $\Theta\in \wt{\cC}(\cH)$.
\item[(iii)] $A_\Theta$ is symmetric if and only if $\Theta$ is symmetric and ${\rm n}_\pm(A_\Theta) = {\rm n}_\pm(\Theta)$ holds. In particular, $A_\Theta$ is self-adjoint if and only if $\Theta$ is self-adjoint.
\item[(iv)] If $A_\Theta = A_\Theta^\ast$ and $A_{\wt\Theta} = A_{\wt\Theta}^\ast$, then for every $p\in(0,\infty]$ the following equivalence holds
\[ 
          (A_{\Theta}-\I)^{-1} - (A_{\wt{\Theta}}-\I)^{-1}\in{\mathfrak S}_p(\gH) \ 
\Longleftrightarrow \          
                  (\Theta-\I)^{-1} - (\widetilde{\Theta}-\I)^{-1}\in{\mathfrak S}_p(\cH).
\]
      If additionally $\dom(\Theta) = \dom(\widetilde{\Theta})$, then 
\[
         \overline{\Theta - \widetilde{\Theta}} \in \mathfrak{S}_p(\cH)\ \Longrightarrow \ 
        (A_{\Theta}-\I)^{-1} - (A_{\wt{\Theta}}-\I)^{-1}\in{\mathfrak S}_p(\gH).
        \]    
\end{itemize} 
\end{theorem}

Notice that according to \eqref{eq:ThetaResolv}, the deficiency indices of a symmetric linear relation $\Theta$ can be defined as the deficiency indices of its operator part $\Theta_{\op}$. Moreover, a self-adjoint linear relation $\Theta$ is said to belong to the von Neumann--Schatten ideal $\gS_p$ if its operator part $\Theta_{\op}$ belongs to $\gS_p(\cH_{\op})$.

\begin{remark}
The proof of Theorem \ref{th:properext}(i)--(ii) can be found in, e.g., \cite[Prop.~7.8]{DM17}, \cite[Prop.~14.7]{schm}; (iii) was obtained in \cite[Prop.~3]{mmm92b}, see also \cite[Prop.~7.14]{DM17}; for the proof of item (iv) see \cite[Theorem 2]{DM91}. 
\end{remark}

\subsection{Weyl functions and extensions of semibounded operators}\label{ss:a3}

With every boundary triplet $\Pi = \{\cH,\Gamma_0,\Gamma_1\}$ one can associate two linear operators
\begin{align*}
A_0 & :=A^* \upharpoonright\ker(\Gamma_0), & A_1 & := A^*\upharpoonright\ker(\Gamma_1).
\end{align*}
Clearly, \eqref{eq:ATheta} implies $A_0 = A_{\Theta_0}$ and $A_1 = A_{\Theta_1}$, where $\Theta_0 = \{0\}\times \cH$ and $\Theta_1 = \cH\times\{0\}$. It easily follows from Theorem \ref{th:properext}(iii) that $A_0 = A_0^\ast$ and $A_1 = A_1^\ast$.

\begin{definition}[{\cite{DM91}}]\label{def_Weylfunc}
Let  $\Pi=\{\cH,\gG_0,\gG_1\}$ be a boundary triplet for $A^*$.
The operator-valued function 
$M\colon \rho(A_0)\to  \cB(\cH)$ defined by
\begin{equation}\label{II.1.3_01}
M(z):=\Gamma_1(\Gamma_0\upharpoonright\cN_z)^{-1}, \qquad
z\in\rho(A_0),
\end{equation}
is called  {\em the Weyl function} corresponding to the boundary triplet $\Pi$.
\end{definition}

The Weyl function is well defined and holomorphic on $\rho(A_0)$.
Moreover, it is a Herglotz--Nevanlinna function  (see \cite[\S 1]{DM91}, \cite[\S 7.4.2]{DM17} and also \cite[\S 14.5]{schm}).

Assume now that $A\in \cC(\gH)$ is a lower semibounded operator, i.e., $A\ge a\,{\rm
I}_{\gH}$ with some $a\in\R$. Let $a_0$ be  the largest lower bound for $A$, 
\[
a_0 := \inf_{f\in \dom(A)\setminus\{0\}}\frac{(Af,f)_\gH}{\|f\|_\gH^2}.
\]
The Friedrichs extension of $A$ is denoted by $A_F$. If $\Pi =
\{\cH,\Gamma_0,\Gamma_1\}$ is a boundary triplet for $A^*$ such that $A_0=A_F$, then the corresponding
Weyl function $M$ is holomorphic on $\C\setminus [a_0,\infty)$. 
Moreover, $M$ is strictly increasing on $(-\infty,a_0)$ (that is, for all $x$, $y\in (-\infty,a_0)$, $M(x)-M(y)$ is positive definite whenever $x>y$) and the following strong resolvent limit exists (see \cite{DM91}) 
   \be\label{eq:srMa}
M(a_0):=s-R-\lim_{x\uparrow a_0}M(x).
  \ee
However,  $M(a_0)$ is in general a closed linear relation which is
bounded from below.

\begin{theorem}[\cite{DM91, mmm92}]\label{th:lsbTheta}
Let $A\ge a\,{\rm I}_{\gH}$ with some $a\ge 0$ and let $\Pi = \{\cH,\Gamma_0,\Gamma_1\}$
be a boundary triplet for $A^*$ such that $A_0=A_F$. Also, let $\Theta=\Theta^\ast
\in\wt{\cC}(\cH)$ and $A_\Theta$ be the corresponding self-adjoint extension \eqref{eq:ATheta}. If
$M(a)\in \cB(\cH)$, then:
\begin{itemize}
\item[(i)] $A_\Theta \ge a\,{\rm I}_{\gH}$ if and only if $\Theta-M(a) \ge \bO_{\cH}$.
\item[(ii)]
\[
\kappa_-(A_\Theta - a\,{\rm I}) = \kappa_-(\Theta-M(a)).
\]
\end{itemize}
If additionally $A$ is positive definite, that is, $a>0$, then:
\begin{itemize}
\item[(iii)] $A_\Theta$ is positive definite if and only if $\Theta(0):=\Theta-M(0)$ is positive definite.
\item[(iv)] For every $p\in (0,\infty]$ the following equivalence holds
\[
A_{\Theta}^- \in \mathfrak{S}_p(\gH) \ \Longleftrightarrow \ \Theta(0)^- \in \mathfrak{S}_p(\cH),
\]
where $\Theta(0)^-:= \Theta(0)_{\op}^-\oplus \Theta(0)_{\infty}$. 
\item[(v)] For every $\gamma\in (0,\infty)$ the following equivalence holds
\[
\lambda_j(A_{\Theta}) =  j^{-\gamma}(a+o(1)) \ \Longleftrightarrow \  \lambda_j({\Theta(0)})= j^{-\gamma}(b+o(1))
\]
as $j\to \infty$. Moreover, either $ab\not = 0$ or $a = b = 0$.
\end{itemize}
\end{theorem}

\begin{remark}
For the proofs of (i) and (ii) consult Theorems~5 and 6 in \cite{DM91}; the proofs of (iii)--(v) can be found in \cite[Theorem~3]{mmm92}.
\end{remark}

We also need the following important statement (see \cite[Theorem~3]{DM91} and \cite[Theorem~8.22]{DM17}).

\begin{theorem}[\cite{DM91}]\label{th:lsb=lsb}
Assume the conditions of Theorem \ref{th:lsbTheta}. Then the following statements
\begin{itemize}
\item[(i)] $\Theta\in\wt{\cC}(\cH)$ is lower semibounded,
\item[(ii)] $A_\Theta$ is lower semibounded,
\end{itemize}
are equivalent if and only if $M(x)$ tends uniformly to $-\infty$ as $x\to -\infty$, that is, for every $N>0$ there exists $x_N<0$ such that $M(x)< -N\cdot {\rm I}_{\cH}$ for all $x<x_N$.
\end{theorem}

The implication $(ii)\Rightarrow (i)$ always holds true (cf. Theorem \ref{th:lsbTheta}(i)), however, the validity of the converse implication requires that $M$ tends uniformly to $-\infty$. Let us mention in this connection that the weak convergence of $M(x)$ to $-\infty$, i.e., the relation
\[
\lim_{x\to-\infty}(M(x)h, h)_\cH = -\infty 
\]
holds for all $ h\in \cH\setminus \{0\}$ whenever $A_0 = A_F$. Moreover, this relation  characterizes Weyl
functions of the Friedrichs extension $A_F$ among  all non-negative (and even lower semibounded) self-adjoint extensions of $A$ (see \cite{KreOvc78},  \cite[Proposition 4]{DM91}).

The next new result establishes a connection between the essential spectra of $A_\Theta$ and $\Theta$ and also it can be seen as an improvement of Theorem \ref{th:lsbTheta} (iv). 

\begin{theorem}\label{th:ess-impl-ess}
Let $A\ge a_0\, I_{\gH}>0$  and let $\Pi =
\{\cH,\Gamma_0,\Gamma_1\}$ be a boundary triplet for $A^*$ such that $A_0=A_F$. Also, let $M$ be the corresponding Weyl function and let $\Theta=\Theta^\ast\in \wt{\cC}(\cH)$ be such that $A_\Theta=A_\Theta^\ast$ is lower semibounded. Then the following equivalences  hold:
   \begin{align}
\inf \sigma_{\ess}(A_\Theta)\ge 0\ \Longleftrightarrow \ \inf\sigma_{\ess}(\Theta - M(0))\ge 0,  \label{ess-B=ess-AB} \\
\inf \sigma_{\ess}(A_\Theta)>0\ \Longleftrightarrow \ \inf\sigma_{\ess}(\Theta - M(0))>0,  \label{ess-B-impl-ess-AB} \\
\inf\sigma_{\ess}(A_\Theta)=0 \ \Longleftrightarrow  \ \inf \sigma_{\ess}(\Theta - M(0))=0. \label{ess-B-impl-ess-AB2}
   \end{align}
\end{theorem}

  \begin{proof}
First observe that \eqref{ess-B=ess-AB} easily follows from Theorem \ref{th:lsbTheta}(iv). Hence it remains to prove \eqref{ess-B-impl-ess-AB} since \eqref{ess-B-impl-ess-AB2} follows from \eqref{ess-B=ess-AB} and \eqref{ess-B-impl-ess-AB}.

Since $A$ is uniformly positive and $A_0=A_F$, we can assume without loss of generality that  $M(0) =\bO_\cH$. Indeed, $M(0)\in \cB(\cH)$ and hence we can replace the  boundary triplet $\Pi = \{\cH,\Gamma_0,\Gamma_1\}$  by the triplet $\Pi_0
= \{\cH,\Gamma_0,\Gamma_1 - M(0)\Gamma_0\}$ and in this case the  Weyl function
$M(\cdot)$ and the boundary relation $\Theta$ are replaced respectively by $M(\cdot) - M(0)$ and  $\Theta -
M(0)$. Moreover, for simplicity we shall assume that $\Theta=B\in \cC(\cH)$ is a self-adjoint linear operator (cf.~\eqref{eq:ThetaResolv}). 

We divide the proof of \eqref{ess-B-impl-ess-AB} into two parts. 

(i) Let us first establish the implication ``$\Leftarrow$"\ in \eqref{ess-B-impl-ess-AB}. %
 For $a := \inf\sigma_{\ess}(B) >0$, we set 
\begin{align*}
 \cH_1 & :=\ran E_B\big([a,\infty)\big), & \cH_2 & :=\ran E_B\big((-\infty,a)\big) = \cH_1^\perp,
\end{align*}
and then define the operators $B_j := B \upharpoonright \cH_j$, $j\in\{1,2\}$. 
Since both subspaces $\cH_1$  and $\cH_2$ are reducing for $B$,  $B = B_1\oplus B_2$. Moreover, we set
   \begin{equation}\label{aux-oper_wt-B1}
\wt B:=B_1\oplus a {\rm I}_{\cH_2}\ge a {\rm I}_{\cH}>0.
  \end{equation}
Combining this inequality with the assumption $M(0) = \bO_\cH$ and applying Theorem
\ref{th:lsbTheta}(iii), we obtain that $A_{\wt B}\ge \ti{a}\,{\rm I}_\gH$ for some $\ti{a} > 0$. 

On the other hand, $B$ is lower semibounded since so is $A_B$ (see a remark after Theorem \ref{th:lsb=lsb}). Hence the operator $B_2$ is lower
semibounded  too and by the definition of $B_2$ either $B_2$ is
finite-rank or the point $a$ is the only accumulation point for $\sigma(B_2)$, i.e.,
$(B_2 - a\,{\rm I}_{\cH_2}) \in{\gS}_{\infty}(\cH_2)$. Therefore,
  \begin{equation}\label{3}
B - \wt B =  
\bO_{\cH_1}\oplus (B_2 -
a\,{\rm I}_{\cH_2}) \in{\gS}_{\infty}(\cH).
 \end{equation}
By Theorem \ref{th:properext} (iv), this relation yields
  \begin{equation}\label{A.16}
(A_{B}-\I)^{-1} - (A_{\wt B}-\I)^{-1} \in{\gS}_{\infty}(\gH),
\end{equation}
which, in turn, implies $\sigma_{\ess}(A_B) = \sigma_{\ess}(A_{\wt B})$. Hence
   \begin{equation}\label{infess-B-impl-infess-AB2}
\inf\sigma_{\ess}(A_B) = \inf \sigma_{\ess}(A_{\wt B}) \ge \ti{a}>0.
   \end{equation}
This proves the implication ``$\Leftarrow$"\, in \eqref{ess-B-impl-ess-AB}. 

(ii) To prove the remaining implication ``$\Rightarrow$"\, in \eqref{ess-B-impl-ess-AB}, let  $b := \inf\sigma_{\ess}(A_B) >0$ and assume the contrary, that is $a = \inf\sigma_{\ess}(B)
\le 0$. Then at least one of the following two
conditions is satisfied: 
\begin{align*}
\dim\ran E_{B}\big((-\infty,0)\big) & = \infty, & 
\dim\ran E_{B}\big([0, \delta)\big) & = \infty \ \ \text{for all}\ \ \delta>0.
\end{align*}
In the first case, Theorem \ref{th:lsbTheta}(ii) implies $\kappa_-(A_B) = \kappa_-(B)  = \infty$.  Since $A_B$ is lower semibounded, we get  $b=\inf\sigma_{\ess}(A_B)\le 0$, which contradicts the assumption $b>0$.

In the second case, recall that $A\ge a_0\, I_{\gH}$ with $a_0>0$. The corresponding Weyl function $M$ is analytic on $(-\infty,a_0)$ and $M(x)=M(x)-M(0)$ is positive definite for all $x\in (0,a_0)$. Fix some $x\in (0,a_0\wedge b)$ and let $\varepsilon>0$ be such that $M(x)\ge \varepsilon\, {\rm I}_{\cH}$. Noting that 
\[
(Bf,f)_{\cH} < \delta \|f\|^2_{\cH}
\]
for all $f\in \ran E_{B}([0, \delta)\big)\setminus\{0\}$, we get 
\[
((B-M(x))f,f)_{\cH} <(\delta-\varepsilon)\|f\|_\cH^2< 0
\]
for all $f\in \ran E_{B}([0, \delta))\setminus\{0\}$ whenever $\delta<\varepsilon$. By  Theorem \ref{th:lsbTheta}(ii),   
\[
\kappa_-(A_B -x) = \kappa_-(B - M(x)) = \infty,
\]
 and hence $\inf\sigma_{\ess}(A_B)\le x<b$ since $A_B$ is lower semibounded. This contradiction finishes the proof.
   \end{proof}

\subsection{Direct sums of boundary triplets}\label{ss:a4}

Let $J$ be a countable index set, $\# J= \aleph_0$. For each $j\in J$, let $A_j$ be a closed densely defined symmetric operator in a separable Hilbert space $\gH_j$ such that $0<\mathrm{n}_+(A_j)=\mathrm{n}_-(A_j)\leq \infty$. Also, let $\Pi_j=\{\cH_j,\gG_{0,j},\gG_{1,j}\}$ be a boundary triplet for the operator $A_j^*$, $j\in J$. In the Hilbert space $\gH:=\oplus_{j\in J}\gH_j$, consider the operator $A:=\oplus_{j\in J} A_j$, which is symmetric and $\mathrm{n}_+(A)=\mathrm{n}_-(A) = \infty$. Its adjoint is given by $A^\ast=\oplus_{j\in J} A_j^\ast$. Let us define a direct sum $\Pi := \oplus_{j\in J}\Pi_j$ of boundary triplets $\Pi_j$ by setting
\begin{align}\label{pi_plus}
 \cH & = \oplus_{j\in J}\cH_n, & \gG_0 & :=\oplus_{j\in J} \gG_{0,n},\quad \gG_1:=\oplus_{j\in J} \gG_{1,n}.
 \end{align}
Note that $\Pi=\{\cH, \Gamma_0,\Gamma_1\}$ given by \eqref{pi_plus} may not form a boundary triplet for $A^\ast$ in the sense of Definition
\ref{def_ordinary_bt} (for example, $\gG_0$ or $\gG_1$ may be unbounded) and first counterexamples were constructed by A.\ N.\ Kochubei. 
The next result provides several criteria for \eqref{pi_plus}
to be a boundary triplet for the operator $A^\ast=\oplus_{n=1}^\infty A^\ast_n $.

\begin{theorem}[\cite{KM10,MalNei_08,cmp13}]\label{th_bt_criterion}
Let $A=\oplus_{j\in J} A_j$ and let  $\Pi=\{\cH,\gG_0,\gG_1\}$ be defined by
\eqref{pi_plus}. Then the following conditions are equivalent:
\begin{itemize}
\item[(i)] $\Pi=\{\cH,\gG_0,\gG_1\}$ is a boundary triplet for the operator $A^\ast$.
\item[(ii)]
The mappings $\gG_0$ and $\gG_1$ are bounded as mappings from $\dom(A^\ast)$ equipped with the graph norm to $\cH$.
\item[(iii)] The Weyl functions $M_j$ corresponding to the triplets $\Pi_j$, $j\in J$, satisfy the following condition
    \begin{equation}\label{WF_criterion}
\sup_{j\in J}\big( \|M_j(\I)\|_{\cH_j}\vee \|(\im M_j(\I))^{-1}\|_{\cH_j}\big) < \infty.
   \end{equation}
\item[(iv)] If in addition $a\in \R$ is a point of a regular type of the operator $A$ (i.e., there exists a positive constant $c>0$ such that $\|(A-a)f\|\ge c\|f\|$ for all $f\in\dom(A)$), then
(i)--(iii) are further equivalent to
  \begin{equation}\label{III.2.2_02}
\sup_{j \in J}\max\big\{\|M_j(a)\|_{\cH_j},\, \|M'_j(a)\|_{\cH_j},\, \|\bigl(M'_j(a)\bigr)^{-1}\|_{\cH_j}\}<\infty.
     \end{equation}
\end{itemize}
\end{theorem}

Based on these criteria, different regularizations $\wt{\Pi}_j$ of triplets  $\Pi_j$ such that the corresponding direct sum
$\wt{\Pi} = \oplus_{j\in J} \wt{\Pi}_j$ forms a boundary triplet for $A^\ast=\oplus_{j\in J} A_j^\ast$ were suggested in \cite{cmp13, KM10, MalNei_08}.

\bigskip
\noindent
\ack We thank Noema Nicolussi for useful discussions and helpful remarks. We are also grateful to the referees for the careful reading of our manuscript, their remarks and hints with respect to the literature that have helped to improve the exposition.

A.K. appreciates the hospitality at the Department of Theoretical Physics, Nuclear Physics Institute, during several short stays in 2016, where a part of this work was done. 


\end{document}